\documentclass[11pt]{article}
\usepackage[margin=1in]{geometry}

\usepackage[utf8]{inputenc}
\usepackage{amsmath,amssymb,amsfonts,amsthm,bm}
\usepackage{latexsym,bbm,graphicx,float}
\usepackage{mathtools,braket}
\usepackage[multiple]{footmisc} 
\usepackage{enumitem}
\usepackage{multirow, multicol}
\usepackage{thmtools, thm-restate}
\usepackage[plain]{algorithm}
\usepackage{algpseudocode}
\usepackage[usenames,dvipsnames,svgnames,table]{xcolor}
\usepackage{subcaption}
\usepackage{soul}
\usepackage{varwidth}
\usepackage{url}
\usepackage{soul} 
\usepackage{todonotes}
\usepackage{titlesec} 

\usepackage[backend=biber,backref=true,backrefstyle=two,style=alphabetic,maxalphanames=4,maxcitenames=2,maxbibnames=99,natbib=true]{biblatex}
\usepackage[breaklinks,hyperfootnotes=false]{hyperref}
\hypersetup{
	colorlinks   = true, 
	urlcolor     = Maroon, 
	linkcolor    = Maroon, 
	citecolor   =  Maroon 
}

\usepackage{cleveref}

\newtheorem{theorem}{Theorem}

\newtheorem{obs}{Observation}
\newtheorem{definition}{Definition}

\newcommand{\abs}[1]{\left|#1\right|}
\newcommand{\proofcase}[1]{\par\smallskip\noindent{\textbf{#1}}}

\DeclareCaptionSubType*{table}

\DeclareMathOperator*{\argmax}{arg\,max}

\algnewcommand\algprocedure{\textbf{Procedure:}}
\algnewcommand\Procedurename{\item[\underline{\algprocedure}]}
\algnewcommand\algmain{\textbf{Main:}}
\algnewcommand\Main{\item[\underline{\algmain}]}
\algnewcommand\algorithmicinput{\textbf{Input:}}
\algnewcommand\Input{\item[\algorithmicinput]}
\algnewcommand\algorithmicoutput{\textbf{Output:}}
\algnewcommand\Output{\item[\algorithmicoutput]}

\makeatletter
\renewenvironment{proof}[1][\proofname] {\pushQED{\qed}\normalfont\topsep\z@\@plus0\p@\relax\trivlist\item[\hskip\labelsep\bfseries#1\@addpunct{:}]\ignorespaces}{\popQED\endtrivlist\@endpefalse}
\makeatother

\makeatletter
\def\thm@space@setup{%
	\thm@preskip=8pt plus 2pt minus 4pt
	\thm@postskip=\thm@preskip 
}
\makeatother
\makeatletter
\renewcommand{\ALG@beginalgorithmic}{\small}
\makeatother

\titlespacing\section{0pt}{12pt plus 2pt minus 2pt}{-1pt plus1pt minus 1pt}
\titlespacing\subsection{0pt}{12pt plus 2pt minus 2pt}{-1pt plus 1pt minus 1pt}
\titlespacing\subsubsection{0pt}{12pt plus 2pt minus 2pt}{-1pt plus 1pt minus 1pt}

\setlist{topsep = 0pt plus1pt}

\allowdisplaybreaks

\urldef{\email}\path|{vijay.menon, kate.larson}@uwaterloo.ca|

\title{Algorithmic Stability in Fair Allocation of Indivisible Goods Among Two Agents}

 \author{
 	Vijay Menon\footnote{David R.\ Cheriton School of Computer Science, University of Waterloo.\ \email}
 	\and 
 	Kate Larson\footnotemark[1]
 }

\date{}

\begin{document}
	
\def\RankClaimsinApp{1}
\def\rlStabilityClaimsinApp{1}
\def\rlStabilityUBlemmainApp{0}
\def\rlStabilityLBlemmainApp{0}
\def\weakStabilityinApp{0}

\maketitle

\begin{abstract}
		Many allocation problems in multiagent systems rely on agents specifying cardinal preferences. However, allocation mechanisms can be sensitive to small perturbations in cardinal preferences, thus causing agents who make ``small" or ``innocuous" mistakes while reporting their preferences to experience a large change in their utility for the final outcome. To address this, we introduce a notion of algorithmic stability and study it in the context of fair and efficient allocations of indivisible goods among two agents.	We show that it is impossible to achieve exact stability along with even a weak notion of fairness and even approximate efficiency. As a result, we propose two relaxations to stability, namely, approximate-stability and weak-approximate-stability, and show how existing algorithms in the fair division literature that guarantee fair and efficient outcomes perform poorly with respect to these relaxations. This leads us to explore the possibility of designing new algorithms that are more stable. Towards this end, we present a general characterization result for pairwise maximin share allocations, and in turn use it to design an algorithm that is approximately-stable and guarantees a pairwise maximin share and Pareto optimal allocation for two agents. Finally, we present a simple framework that can be used to modify existing fair and efficient algorithms in order to ensure that they also achieve weak-approximate-stability.
\end{abstract}

\section{Introduction}

There are many scenarios where agents are assumed to have cardinal preferences. For instance, this is crucial in the rent-division setting \cite{gal17} and is typical in the vast literature on fair allocation of both divisible \cite{brams96b} and indivisible goods \cite{lip04}. While one could argue that assuming cardinal preferences is reasonable in such contexts, there is an implicit assumption that agents are capable of reporting preferences accurately.  However,  it is not hard to imagine scenarios where many would find it hard to provide exact numerical values to our preferences.

To elaborate on this further, consider the well-studied problem of allocating $m$ indivisible goods among $n \geq 2$ agents. A popular algorithm for this is the maximum Nash welfare (MNW) solution which computes an allocation that is Pareto optimal (PO) and satisfies a fairness notion called EF1 \cite{car16}. Informally, in an allocation that satisfies EF1 an agent $i$ does not envy agent $j$ after she removes some good from $j$'s bundle, whereas Pareto optimality of an allocation implies that there is no other allocation where every agent receives at least as much utility and at least one of the agents strictly more. Given our current state of knowledge on fair and efficient allocations, the MNW solution essentially provides the best-known guarantees. However, as we will soon see, there is at least one aspect with respect to which it is lacking. The issue we will discuss is not specific to MNW but is something that can be raised with respect to different algorithms that assume access to cardinal preferences in different settings. Nevertheless, we use MNW here since our motivation to look at this issue stemmed from observing  examples on Spliddit (\url{www.spliddit.org}), a popular fair division website which uses the MNW solution \cite{gold15}. 
To illustrate our concern, consider an example with two agents (A and B) and four goods ($g_1,\ldots, g_4$). The agents have additive valuations---meaning, value of a set of goods is the sum of values of each of the goods---and their values for the goods are in \Cref{tab1a}. Spliddit uses the MNW solution to compute an EF1 and PO allocation for this instance, and it returns one where agent A receives $\{g_2, g_4\}$ and agent B receives $\{g_1, g_3\}$. What if, however, agent A made a minor `mistake' while reporting the values and instead reported the values in \Cref{tab1b}? Note that the two valuations are almost identical, with the value of each good being off by at most 2. Therefore, intuitively, it looks like we would, ideally, like to have similar outputs---and more so since the allocation for the original instance satisfies EF1 and PO even with respect to this new instance. However, does the MNW solution do this? No, and in fact the allocation returned in this case is one where agent A gets $g_4$ alone and agent B gets the rest. This in turn implies that agent A is losing roughly 38\% of their utility for the `mistake', which seems highly undesirable.

\begin{table}[!tb]    
    \begin{subtable}{.5\linewidth}
      \centering       
        \begin{tabular}{l l|l}
          & A & B \\
	  \cline{2-3}
	  & & \\[-2.5ex]
	 $g_1$ &  104 & 162 \\	  
	 $g_2$ &  273 & 250 \\
	 $g_3$ & 186 & 240\\
	 $g_4$&  437 & 348
        \end{tabular}
         \subcaption{Original instance} \label{tab1a}
    \end{subtable}%
    \begin{subtable}{.5\linewidth}
      \centering
        \begin{tabular}{l l|l}
          & A & B \\
	  \cline{2-3}
	  & & \\[-2.5ex]
	 $g_1$ &  105 & 162 \\	  
	 $g_2$ &  271 & 250 \\
	 $g_3$ & 186 & 240\\
	 $g_4$&  438 & 348
        \end{tabular}
	\subcaption{Instance where agent A makes minor mistakes} \label{tab1b}
    \end{subtable}  
\end{table}

The example described above is certainly not one-off, and in fact we will show later how there are far worse examples for different fair division algorithms. More broadly, we believe that this is an issue that can arise in many problems where the inputs are assumed to be cardinal values. After all, it is not hard to imagine scenarios where many of us may find it hard to convert our preferences to precise numerical values. Now, of course, it is easy to see that if we insist on resolving this issue completely---meaning, if we insist that the agent making the `mistake' should not experience any change in their outcome---then it cannot be done in any interesting way as long as we allow the `mistakes' to be arbitrary and also insist that the algorithm be deterministic.\footnote{Although there is work on randomized fair allocations (e.g., \cite{bogo01, budi13}), as pointed out by \citet{car16}, randomization is not appropriate for many practical fair division settings where the outcomes are just used once.} However, it is possible to impose some structure on the `mistakes' made by agents. For example, in many settings, it might be reasonable to assume that agents are able to, at a minimum, maintain the underlying ordinal structure of their preferences. That is, if an agent considers good $g$ to be the $r$-th highest valued good according to their true preference, then this information is also maintained in the `mistake'. Note that this is indeed the case in the example above, and so, more broadly, this is the setting we consider. Our goal in this paper is to try and address the issue that we observe in the example, and intuitively we want to design algorithms where an agent does not experience a large change in their utility as long as their report is only off by a little.      

Before we make this more concrete, a reader who is familiar with the algorithmic game theory (AGT) literature might have the following question: ``Why not just consider ordinal algorithms?'' After all, these algorithms will have the property mentioned above when the underlying ordinal information is maintained, and moreover there is a body of literature that focuses on designing algorithms that only use ordinal information and still provide good guarantees with respect to the underlying cardinal values (e.g., see \cite{bout15,ansh16,ansh17,goel17,abra18}). Additionally, and more specifically in the context of fair allocations, there is also a line of work that considers ordinal algorithms \cite{bouv10,brams14,aziz15b,segal17}. While this is certainly a reasonable approach, there are a few reasons why this is inadequate: \emph {i)} Constraining algorithms to use only ordinal information might be too restrictive. In fact, this is indeed the case here since we show that there are no ordinal algorithms that are EF1 and even approximately PO. Additionally, assuming that the agents only have ordinal preferences might be too pessimistic in certain situations.\ \emph{ii)} There are systems like Spliddit that are used in practice and which explicitly elicit cardinal preferences, and so we believe that the approach here will be useful in such settings. 

Given this, we believe that there is need for a new notion to address this issue. We term this \textit{stability}, and informally our notion of stability captures the idea that the utility experienced by an agent {should not change much} as long as they make {``small'' or ``innocuous'' mistakes when reporting their preferences}.\footnotemark~Although the general idea of algorithmic stability is certainly  not new (see \Cref{sec:rw} for a discussion), to best of our knowledge, the notion of stability we introduce here (formally defined in \Cref{sec:stability}) has not been previously considered. Therefore, we introduce this notion in the context of problems where cardinal preferences are elicited, and explicitly 
advocate for it to be considered during algorithm or mechanism design. This in turn constitutes what we consider as the main contribution of this paper. We believe that if the algorithms for fair division---and in fact any problem where cardinal preferences are elicited---are to be truly useful in practice they need to have some guarantees on stability, and so towards this end we consider the problem of designing stable algorithms in the context of fair allocation of indivisible goods among two agents.
\footnotetext{In the AGT literature, the term \textit{stability} is usually used in the context of stable algorithms in two-sided matching \cite{gale62}. Additionally, the same term is used in many different contexts in the computer science literature broadly (e.g., in learning theory, or when talking about, say, stable sorting algorithms). Our choice of the term stability here stems from the usage of this term in learning theory (see \Cref{sec:rw} for an extended discussion) and therefore should not be confused with the notion of stability in two-sided matching.}

We begin by formally defining the notion of stability and show how one cannot hope for stable algorithms that are EF1 and even approximately PO. As a result, we propose two relaxations, namely, approximate-stability and weak-approximate-stability, and show how existing algorithms that are fair and efficient perform poorly even in terms of these relaxations. This implies that one has to design new algorithms, and towards this end we present a simple, albeit exponential, algorithm for two agents that is approximately-stable and that guarantees pairwise maximin share (PMMS) and PO allocation; the algorithm is based on a general characterization result for PMMS allocations which we believe might be of independent interest. Finally, we show how a small change to the existing two-agent fair division algorithms can get us weak-approximate-stability along with the properties that these algorithms otherwise satisfy.  

\section{Preliminaries} \label{sec:prelims}
Let $[n] = \{1, \ldots, n\}$ denote the set of agents and $\mathcal{G} = \{g_1, \ldots, g_m\}$ denote the set of indivisible goods that needs to divided among these agents. Throughout, we assume that every agent $i \in [n]$ has an additive valuation function $v_i: 2^{\mathcal{G}} \rightarrow \mathbb{Z}_{\geq0}$,\footnote{We assume valuations are integers to model practical deployment of fair division algorithms (e.g., adjusted winner protocol, Spliddit). All the results hold even if we assume that the valuations are non-negative real numbers.} where $\mathcal{V}$ denotes the set of all additive valuation functions, $v_i(\emptyset) = 0$, and additivity implies that for a $S \subseteq \mathcal{G}$ (which we often refer to as a bundle), $v_i(S) = \sum_{g\in S} v_i(\{g\})$. {For ease of notation, we often omit $\{g\}$ and instead just write it as $v_i(g)$.} We 
also assume throughout that $\forall i \in [n]$, $v_i(\mathcal{G}) = T$ for some $T \in  \mathbb{Z}^+$, and that $\forall g\in \mathcal{G}$, $v_i(g) > 0$. Although the assumption that agents have positive value for a good may not be valid in certain situations, in \Cref{sec:qa} we argue why this is essential in order to obtain anything interesting in context of our problem. Finally, for $S \subseteq \mathcal{G}$ and $k \in [n]$, we use $\Pi_k(S)$ to denote the set of ordered partitions of $S$ into $k$ 
bundles, and for an allocation $A \in \Pi_n (\mathcal{G})$, where $A = (A_1, \ldots, A_n)$, use $A_i$ to denote the bundle allocated to agent $i$. 

We are interested in deterministic algorithms $\mathcal{M}: \mathcal{V}^n \to \Pi_n(\mathcal{G})$ that produce fair and efficient (i.e., Pareto optimal) allocations. For $i\in [n]$ and a profile of valuation functions $(v_i, v_{-i}) \in \mathcal{V}^n$, we use $v_i(\mathcal{M}(v_i, v_{-i}))$ to denote the utility that $i$ obtains from the allocation returned by $\mathcal{M}$.  Pareto optimality and the fairness notions considered are defined below. 

\begin{definition}[{$\beta$-Pareto optimality ($\beta$-PO)}] \label{def:po}
	Given a $\beta \geq 1$, an allocation $A \in \Pi_n (\mathcal{G})$ is $\beta$-Pareto optimal if 
	\begin{equation*}
		\forall A' \in  \Pi_n (\mathcal{G}): \; \left(\exists i \in [n], v_i(A_i') > \beta \cdot v_i(A_i) \right)\Rightarrow \left(\exists j \in [n], v_j(A_i') < v_j(A_j)\right).  
	\end{equation*}
\end{definition}

In words, an allocation is $\beta$-Pareto-optimal if no agent can be made strictly better-off by a factor of $\beta$ without making another agent strictly worse-off; Pareto optimality refers to the special case when $\beta = 1$. For allocations $A, A'$, we say that $A'$ $\beta$-Pareto dominates $A$ if every agent receives at least as much utility in $A'$ as in $A$ and at least one agent is strictly better-off by a factor of $\beta$ in $A'$.

We consider several notions of fairness, namely, pairwise maximin share (PMMS), envy-freeness up to least positively valued good (EFX), and envy-freeness up to one good (EF1). Among these, the main notion we talk about here (and design algorithms for) is PMMS, which is the strongest and which implies the other two notions. Informally, an allocation $A$ is said to be a PMMS allocation if every agent $i$ is guaranteed to get a bundle $A_i$ that she values more than the bundle she receives when she plays \textit{cut-choose} with any other agent $j$ (i.e., she partitions the combined bundle of her allocation and the allocation $A_j$ of $j$ and receives the one she values less).

\begin{definition}[{pairwise maximin share} (PMMS)]
	An allocation $A \in \Pi_n (\mathcal{G})$ is a pairwise maximin share (PMMS) allocation if 
	\begin{equation*}
		\forall i, \forall j \in [n]: \; v_i(A_i) \geq \max_{B \in \Pi_2(A_i \cup A_j)} \min \{v_i(B_1), v_i(B_2)\}.
	\end{equation*}
\end{definition}

Next, we define EFX and EF1, which as mentioned above are weaker than PMMS. 

\begin{definition}[{envy free up to any positively valued good} (EFX)] \label{def:efx}
	An allocation $A \in \Pi_n (\mathcal{G})$ is envy-free up to any positively valued good (EFX) if 
	\begin{equation*}
		\forall i, \forall j \in [n], \forall g \in A_j \text{ with } v_i(g) > 0: \; v_i(A_i) \geq v_i(A_j \setminus \{g\}).
	\end{equation*}
\end{definition}
In words, an allocation is said to be EFX if agent $i$ is no longer envious after removing any positively valued good from agent $j$'s bundle. 

\begin{definition}[{envy free up to one good} (EF1)] \label{def:ef1}
	An allocation $A \in \Pi_n (\mathcal{G})$ is envy-free up to one good (EF1) if 
	\begin{equation*}
		\forall i, \forall j \in [n], \exists g \in A_j: \; v_i(A_i) \geq v_i(A_j \setminus \{g\}).
	\end{equation*}
\end{definition}
In words, an allocation is said to be EF1 if agent $i$ is no longer envious after removing some good from agent $j$'s bundle. 

In addition to the requirement that the algorithms be fair and efficient, we would also ideally like it  to be stable. We define the notion of stability (and its relaxations) in the next section.

\subsection{Stability} \label{sec:stability}
At a high-level, our notion of stability captures the idea that an agent should not experience a large change in utility as long as they make {``small" or ``innocuous" mistakes} while reporting their preferences. Naturally, a more formal definition requires us to first define what constitutes a `mistake' and what we mean by {``small" or ``innocuous" mistakes} in the context of fair division. So, below, for an $i \in [n]$, $v_i \in \mathcal{V}$, and $\alpha > 0$, we define what we refer to as $\alpha$-neighbours of $v_i$. According to this definition, the closer $\alpha$ is to zero, the ``smaller" is the `mistake', since smaller values of $\alpha$ indicate that agent $i$'s report $v_i'$ (which in turn is the `mistake') is closer to their true valuation function $v_i$. 

\begin{definition}[$\alpha$-neighbours of $v_i$ ($\alpha$-$N(v_i)$)] \label{def:alphaN}
	For $i \in [n]$, $T \in \mathbb{Z}^+$, $\alpha > 0$, and $v_i \in \mathcal{V}$, define $\alpha$-$N(v_i)$, the set of $\alpha$-neighbouring 
	valuations of $v_i$, to be the set of all $v_i'$ such that 
	\begin{itemize}
		\item $\sum_{g \in \mathcal{G}} v'(g) = T$ $($i.e., $v_i(\mathcal{G}) = v'_i(\mathcal{G}))$,
		\item $\forall g, \forall g' \in \mathcal{G}$, $v(g) \geq v(g') \Leftrightarrow v'(g) \geq v'(g')$ (i.e., the ordinal information over the singletons is maintained), and
		\item $\left\|v_i'-v_i\right\|_1 = \sum_{g \in \mathcal{G}} |v'(g) - v(g)| \leq \alpha$.\footnote{We present our results using the $L_1$-norm, although qualitatively they do not change if we use, say, the $L_{\infty}$-norm.}
	\end{itemize}
\end{definition}

Throughout, we often refer to the valuation function $v_i$ of agent $i$ as its \textit{true valuation function} or \textit{true report} and, for some  $\alpha > 0$,  $v_i' \in \alpha$-$N(v_i)$ as its \textit{mistake} or \textit{misreport}. Note that although \textit{true valuation function}, \textit{true report}, and \textit{misreport} are terms one often finds in the mechanism design literature which considers strategic agents, we emphasize that here we are not talking about strategic agents, but about agents who are just unsophisticated in the sense that they are unable to accurately convert their preferences into cardinal values. Also, although we write $\alpha > 0$, it should be understood that in the context of our results here, since the valuation functions of the agents are integral, the only valid values of $\alpha$ are when it is an integer. We use this notation for ease of exposition and because all our results  hold even if we instead use real-valued valuations.  

Now that we know what constitutes a mistake, we can define our notion of stability. 

\begin{definition}[$\alpha$-stable algorithm]
	For $\alpha > 0$, an algorithm $\mathcal{M}$ is said to be $\alpha$-stable if $\forall i \in [n], \forall v_i, \forall v_{-i}$, and $\forall v'_i \in 
	\alpha$-$N(v_i)$,
	\begin{equation} \label{eq:def1}
		v_i\left(\mathcal{M}(v_i, v_{-i})\right) = v_i\left(\mathcal{M}(v'_i, v_{-i})\right).
	\end{equation}
\end{definition}

In words, for an $\alpha > 0$, an algorithm is said to be $\alpha$-stable if for every agent $i \in [n]$ with valuation function $v_i$ and all possible reports of the other agents, the utility agent $i$ obtains is the same when reporting $v_i$ and $v'_i$.
Given the definition above, we have the following for what it means for an algorithm to be \textit{stable}. 

\begin{definition}[stable algorithm]
	An algorithm $\mathcal{M}$ is stable if it is $\alpha$-stable for all $\alpha > 0$.
\end{definition}

Although the ``for all'' in the definition above might seem like a strong requirement at first glance, it is not, for one can easily show that the following observation holds. The proof of this appears in \Cref{app:sec:prelims:proofs}.
\begin{restatable}{obs}{obsStable} \label{obs1}
	An algorithm $\mathcal{M}$ is stable if and only if there exists an $\alpha > 0$ such that $\mathcal{M}$ is $\alpha$-stable. 
\end{restatable}

%
%

It is important to note that the definition for an $\alpha$-stable algorithm is only saying that the utility agent $i$ obtains (and not the allocation itself) when moving from $v_i$ to $v_i'$ is the same. Additionally, although the notion of stability in general may look too strong, it is important to note that there are several algorithms (e.g., the well-known EF1 \textit{draft} mechanism \cite{car16}) that satisfy this definition. In particular, one can immediately see from the definition of stability that every ordinal fair division algorithm---i.e., an algorithm that produces the same output for input profiles $(v_1, \ldots, v_n)$ and $(v'_1, \ldots, v'_n)$ as long as $\forall g, g' \in \mathcal{G}, v_i(g) \geq v_i(g') \Leftrightarrow v'_i(g) \geq v'_i(g')$---is stable.  


%
%

However, in general, and as we will see in \Cref{sec:st-f-e}, the equality in (\ref{eq:def1}) can be too strong a requirement. Therefore, in the next section we propose two relaxations to the strong requirement of stability. 

\subsubsection{Approximate notions of stability} \label{sec:approx-stability}
We first introduce the weaker relaxation which we refer to as weak-approximate-stability. Informally, weak-approximate-stability basically says that the utility that an agent experiences as a result of reporting a neighbouring instance is not too far away from the what would have been achieved if the reports were exact. 

\begin{definition}[$(\epsilon, \alpha)$-weakly-approximately-stable algorithm] \label{def:weak-approx}
	For an $\alpha > 0$ and $\epsilon \geq 1$, an algorithm $\mathcal{M}$ is said to be $(\epsilon, \alpha)$-weakly-stable if $\forall i \in [n], \forall v_i, \forall v_{-i}$, 
	and $\forall v'_i \in \alpha$-$N(v_i)$,
	\begin{equation}    
		\frac{1}{\epsilon} \leq \frac{v_i\left(\mathcal{M}(v'_i, v_{-i})\right)}{v_i\left(\mathcal{M}(v_i, v_{-i})\right)} \: \leq \:  \epsilon.
	\end{equation}
\end{definition}

Although the definition above might seem like a natural relaxation of the notion of stability, as it will become clear soon, it is a bit weak. Therefore, below we introduce the stronger notion which we refer to as approximate-stability. However, before this, we introduce the following, which, for a given valuation function $v$, defines the set, $\text{equiv}(v)$, of valuation functions $v'$ such that the ordinal information over the bundles is the same in both $v$ and $v'$.

\begin{definition}[$\text{equiv}(v)$] \label{def:equiv}
	For a valuation function $v: 2^{\mathcal{G}} \rightarrow \mathbb{Z}_{\geq0}$, $\text{equiv}(v)$ refers to the set of all valuation functions $v'$ such that for all $S_1, S_2 \subseteq \mathcal{G}$,	
	\begin{align*}
		v(S_1) \geq v(S_2) \Leftrightarrow v'(S_1) \geq v'(S_2).
	\end{align*}
\end{definition}
In words, equiv($v$) refers to set all of valuation functions $v'$ such that $v$ and $v'$ induce the same weak order over the set of all bundles (i.e., over the set $2^{\mathcal{G}}$).  Throughout, for an $i \in [n]$, we say that two instances (or profiles) $(v_i, v_{-i})$ and $(v'_i, v_{-i})$ are equivalent if $v_i' \in \text{equiv}(v_i)$. Also, we say that $v_i$ and $v_i'$ are \textit{ordinally equivalent} if $v_i' \in \text{equiv}(v_i)$. 

Equipped with this notion, we can now define approximate-stability. Informally, an algorithm is approximately-stable if it is weakly-approximately-stable and if with respect to every instance that is equivalent to the true reports, it is stable. 
\begin{definition}[$(\epsilon, \alpha)$-approximately-stable algorithm] \label{def:approx-stable}
	For an $\alpha > 0$ and $\epsilon \geq 1$, an algorithm $\mathcal{M}$ is said to be $(\epsilon, \alpha)$-approximately-stable if $\forall i \in [n], \forall v_i, \forall v_{-i}$, 
	\begin{itemize}
		\item $\forall v'_i \in \text{equiv}(v_i)$, $v_i\left(\mathcal{M}(v_i, v_{-i})\right) = v_i\left(\mathcal{M}(v'_i, v_{-i})\right)$, and
		\item $\mathcal{M}$ is $(\epsilon, \alpha)$-weakly-approximately-stable.
	\end{itemize}
\end{definition}
Note that when $\epsilon = 1$ the definitions for both the relaxations (i.e., weak-approximate-stability and approximate-stability) collapse to the one for $\alpha$-stable algorithms. Also, throughout, we say that an algorithm is $\epsilon$-approximately-stable for $\alpha \leq K$ if it is $(\epsilon, \alpha)$-approximately-stable for all $\alpha \in (0,K]$ (similarly for $\epsilon$-weakly-stable).

Although the definition of approximate-stability might be seem a bit contrived at first glance, it is important to note that this is not the case. The requirement that the algorithm be stable on equivalent instances is natural because of the following observation which states that with respect to all the notions that we talk about here any algorithm that satisfies such a notion can always output the same allocation for two instances that are equivalent.\footnote{Note that \Cref{obs2} is only valid for the exact versions of these notions and not for their approximate counterparts.} 

\begin{obs} \label{obs2}
	Let $P$ be a property that is one of EF1, EFX, PMMS, or PO. For all $i \in [n]$, if an allocation $(A_1, \ldots, A_n)$ satisfies property $P$ with respect to the profile $(v_i, v_{-i})$, then $(A_1, \ldots, A_n)$ also satisfies property $P$ with respect to the profile $(v'_i, v_{-i})$, where $v_i' \in \text{equiv}(v_i)$. 
\end{obs}
\begin{proof}
	Consider an arbitrary agent $i \in [n]$, and recall from \Cref{def:equiv} that for any $v_i' \in \text{equiv}(v_i)$, and for any two sets $S_1, S_2 \subseteq \mathcal{G}$, $v_i'(S_1) \geq v_i'(S_2)$ if and only $v_i(S_1) \geq v_i(S_2)$. Given this, the observation follows by using the definitions of the stated properties. 
\end{proof}

\subsection{Some Q \& A on assumptions and definitions} \label{sec:qa}

\textbf{Why the assumption of positive values for the goods?} To see why we make this assumption, consider the fair division instance in \Cref{tab-appa}. Next, consider an arbitrary algorithm $\mathcal{M}$ that is EF1 and PO, and let us assume w.l.o.g.\ that the allocation returned by the algorithm is $(\{g_2\}, \{g_1\})$. Now, let us consider the instance in \Cref{tab-appb}. Note that since $\mathcal{M}$ returns an allocation that is PO and EF1, therefore the output with respect to the instance in \Cref{tab-appb} has to be either $(\{g_1\}, \{g_2\})$ or  $(\emptyset, \{g_1, g_2\})$.  
\begin{table}[!htb]    
	\begin{subtable}{.5\linewidth}
		\centering       
		\begin{tabular}{l l|l}
			& A & B \\
			\cline{2-3}
			& & \\[-2ex]
			$g_1$ &  T-1 & T-1 \\	  
			$g_2$ &  1 & 1 
		\end{tabular}
		\caption{Original instance} \label{tab-appa}
	\end{subtable}%
	\begin{subtable}{.5\linewidth}
		\centering
		\begin{tabular}{l l|l}
			& A & B \\
			\cline{2-3}
			& & \\[-2.5ex]
			$g_1$ &  T & T-1 \\	  
			$g_2$ &  0 & 1 
		\end{tabular}
		\caption{Instance where agent A makes a mistake} \label{tab-appb}
	\end{subtable} 
\end{table}

Given this, consider a scenario where the valuation function mentioned in \Cref{tab-appa} is the true valuation $(v_1)$ of agent $A$. When reporting correctly, she receives the good $g_2$. \Cref{tab-appb} shows agent A's misreport $(v'_1)$, where $\alpha = 2$ and in which case she receives the good $g_1$ or none of the goods. Therefore, now, for any algorithm that is EF1 and PO, we either have  $\frac{v_1(g_1)}{v_1(g_2)} = \frac{T}{1}$, or $\frac{v_1(\emptyset)}{v_1(g_2)} = 0$, and both of these are not informative or useful. 

\textbf{Is there any direct connection by the notion of stability and strategyproofness?} At first glance, the definition of stable algorithms (or its relaxations defined in \Cref{sec:approx-stability}) might seem very similar to the definition for (approximately) strategyproof algorithms (i.e., algorithms where truthful reporting is an (approximately) weakly-dominant strategy for all the agents). Although there is indeed some similarity, it is important to note that these are different notions and neither does one imply the other. For instance, there is a stable algorithm that is EF1 (one can easily see that the well-known EF1 \textit{draft} algorithm \cite[Sec.\ 3]{car16} is stable), but there is no strategyproof algorithm that is EF1 \cite[App.\ 4.6]{aman17}.

\section{Related work} \label{sec:rw}
Now that we have defined our notion, we can better discuss related work. There are several lines of research that are related to the topic of this paper. Some of the connections we discuss are in very different contexts and are by themselves very active areas of research. In such cases we only provide some pointers to the relevant literature, citing the seminal works in these areas. 

\textbf{Connections to algorithmic stability, differential privacy, and algorithmic fairness.} Algorithmic stability captures the idea of stability by employing the 
principle that the output of an algorithm should not ``change much'' when a ``small change'' is made to the input. To the best of our knowledge, notions of stability have not been considered in the AGT literature. However, the computational learning theory literature has considered 
various notions of stability and has, for instance, used them to draw connections between the stability of a learning algorithm and its ability to generalize (e.g., 
\cite{bous02,shalev10}). Although the notion of stability we employ here is based on the same principle, it is defined differently from the ones in this literature. Here we are concerned about the change in utility than an agent experiences when she perturbs her input and deem an algorithm to be approximately-stable if this change is small.

Algorithmic stability can, in turn, be connected to differential privacy \cite{dwork06a,dwork06b}. Informally, differential privacy (DP) requires that the probability of an outcome does not ``change much'' on ``small changes'' to the input. Therefore, essentially, DP can be considered as a notion of algorithmic stability, albeit a very strong one as compared to the ones studied in the learning theory literature (see the discussion in \cite[Sec.\ 1.4]{dwork15}) and the one we consider. In particular, if we were considering randomized algorithms, it is indeed the case that an $\epsilon$-differentially private algorithm is $\text{exp}(\epsilon)$-stable, just like how an $\epsilon$-differentially private algorithm is a ($\text{exp}(\epsilon)-1)$-dominant strategy mechanism \cite{mcsh07}. Nevertheless, we believe that the notion we introduce here is independently useful, and is different from DP in a few ways. First, the motivation is completely different. We believe that our notion may be important even in situations where privacy is not a concern. Second, in this paper we are only concerned with deterministic algorithms and one can easily see that DP is too strong a notion for this case as there are no deterministic and differentially private algorithms that have a range of at least two. 

Finally, the literature on algorithmic (individual-based) fairness captures the idea of fairness by employing the principle that ``similar agents'' should be ``treated similarly'' \cite{dwo12}. Although this notion is employed in contexts where one is talking about two different individuals, note that one way to think of our stability requirement is to think of it as a fairness requirement where two agents are considered similar if and only if they have similar inputs (i.e., say, if one's input is a perturbation of the other's). Therefore, thinking this way, algorithmic fairness can be considered as a generalization of stability, and just like DP it is much stronger and only applicable in randomized settings. In fact, algorithmic fairness can be seen as a generalization of DP (see the discussion in \cite[Sec.\ 2.3]{dwo12}) and so our argument above as to why our notion is useful is relevant even in this case.

\textbf{Connections to robust algorithm and mechanism design.} Informally, an algorithm is said to be robust if it ``performs-well'' even under ``slightly different'' inputs or if the underlying model is different from the one the designer has access to. This notion has received considerable amount of attention in the algorithmic game theory and social choice literature. For instance---and although this line of work does not explicitly term their algorithms as ``(approximately) robust''---the flurry of work that takes the implicit utilitarian view considers scenarios where the agents have underlying cardinal preferences but only provide ordinal preferences to the designer. The goal of the designer in these settings is to then use these ordinal preferences in order obtain an algorithm or mechanism that ``performs well'' (in the approximation sense, with respect to some objective function) with respect to all the possible underlying cardinal preferences (e.g., \cite{bout15,ansh16,ansh17,goel17,abra18}). Additionally, and more explicitly, robust algorithm design has been considered, for instance, in the context of voting (e.g., \cite{shir13,bred17}) and the stable marriage problem \cite{menon18,mai18,chen19}, and robust mechanism design has been considered in the context of auctions \cite{chiesa12,chiesa14,chiesa15} and facility location \cite{menon19}. Although, intuitively, the concepts of robustness and stability might seem quite similar, it is important to note that they are different. Stability requires that the outcome of an algorithm does not ``change much'' if one of the agents slightly modifies its input. Therefore, the emphasis here is to make sure that the outcomes are not very different as long as there is a small change to the input associated with one of the agents. Robustness, on the other hand, requires the outcome of an algorithm to remain ``good'' (in the approximation sense) even if the underlying inputs are different from what the algorithm had access to. Therefore, in this case the emphasis is on making sure that the same output (i.e., one that is computed with the input the algorithm has access to) is not-too-bad with respect to a set of possible underlying true inputs (but ones the algorithm does not have access to). More broadly, one can think of robustness as a feature that a designer aspires to to ensure that the outcome of their algorithm is not-too-bad even if the model assumed, or the input they have access to, is slightly inaccurate, whereas stability in the context that we use here is more of a feature that is in service of unsophisticated agents who are prone to making mistakes when converting their preferences to cardinal values.

\textbf{Related work on fair division of indivisible goods.} The problem of fairly allocating indivisible goods has received considerable attention, with several works proposing different notions of fairness \cite{lip04,bud11,car16,aman18} and computing allocations that satisfy these notions, sometimes along with some 
efficiency requirements \cite{car16,bar18,plaut18,car19}. This paper also studies the problem of computing fair and efficient allocations, but in contrast to previous work our focus is on coming up with algorithms that are also (approximately) stable. While many of these papers address the general case of $n\geq2$ agents, our work focuses on the case of two agents. 

Although a restricted case, the two agent case is an important one and has been explicitly considered in several previous works \cite{brams12,rama13,brams14,vet14,aziz15a,plaut18}. Among these, the work that is most relevant to our results here is that of \citet{rama13}. In particular, and although the results here we derived independently, \citeauthor{rama13}'s paper contains two results that are similar to the ones we have here---first, a slightly weaker version of the $n=2$ case of \Cref{thm:pmmsiff}, and second, a slightly weaker version of \Cref{thm:main}. The exact differences are outlined in Sections~\ref{sec:neccsuff} and~\ref{sec:rl} since we need to introduce a few more notions to make them clear. 

In addition to the papers mentioned above---all of which adopt the model as in this paper where the assumption is that the agents have cardinal 
preferences---there is also work that considers the case when agents have ordinal preferences \cite{bouv10,brams14,aziz15b,segal17}. Although this line of work is 
related in that ordinal algorithms are stable, it is also quite different since usually the goal in these papers is to compute fair allocations if they exist or 
study the complexity of computing notions like possibly-fair or necessarily-fair allocations.

\section{Approximate-Stability in Fair Allocation of Indivisible Goods} 
Our aim is to design (approximately) stable algorithms for allocating a set of indivisible goods among two agents that guarantee pairwise maximin share (PMMS) and Pareto optimal (PO) allocations.  However, before we try to design new algorithms, the first question that arises is: \textit{How do the existing algorithms fare? How stable are they?} We address this below. 

\subsection{How (approximately) stable are the existing algorithms?} \label{sec:existalgo-stable}
We consider the following well-studied algorithms that guarantee PO and at least EF1.
\begin{enumerate}[label=\roman*)]
	\item Adjusted winner protocol \cite{brams96a,brams96b}; returns an EF1 and PO allocation for two agents.
	\item Leximin solution \cite{plaut18}; returns a PMMS and PO allocation for two agents.
	\item Maximum Nash Welfare solution \cite{car16}; returns an EF1 and PO  allocation for any number of agents.
	\item Fisher-market based algorithm \cite{bar18}; returns an EF1 and PO allocation for any number of agents.
\end{enumerate}

All the algorithms mentioned above perform poorly even in terms of the weaker relaxation of stability, i.e., weak-approximate-stability. To see this, consider the instance in \Cref{tab2a}, and note that the allocation that is output by any of these algorithms is agent A getting $g_1$ and agent B getting $g_2$.  Next, consider the instance in \Cref{tab2b}. In this case, if we use any of these algorithms, then the allocation that is output is agent A getting $g_2$ and agent B getting $g_1$. 

Given this, let the values mentioned in \Cref{tab2a} constitute the true valuation function $(v_1)$ of agent $A$. When reporting these, she receives the good $g_1$. \Cref{tab2b} shows agent A's misreport $(v'_1)$ in which case she receives the good $g_2$. Recall from the definition of $\alpha$-neighbours of $v_i$ (\Cref{def:alphaN}) that $v_1' \in \alpha$-$N(v_i)$, where $\alpha = 4$. Therefore, we now have, 
$\frac{v_1(g_2)}{v_1(g_1)} = \frac{1}{T-1}$, 
or in other words, all the four algorithms mentioned above are $(T-1)$-weakly-approximately-stable, even when $\alpha = 4$.

\begin{table}[tb]    
	\begin{subtable}{.5\linewidth}
		\centering       
		\begin{tabular}{l l|l}
			& A & B \\
			\cline{2-3}
			& & \\[-2.5ex]
			$g_1$ &  T-1 & T-2 \\	  
			$g_2$ &  1 & 2 
		\end{tabular}
		\caption{Original instance} \label{tab2a}
	\end{subtable}%
	\begin{subtable}{.5\linewidth}
		\centering
		\begin{tabular}{l l|l}
			& A & B \\
			\cline{2-3}
			& & \\[-2.5ex]
			$g_1$ &  T-3 & T-2 \\	  
			$g_2$ &  3 & 2
		\end{tabular}
		\caption{Instance where agent A makes a mistake} \label{tab2b}
	\end{subtable} 
\end{table}

\subsection{Are there fair and efficient algorithms that are stable?} \label{sec:st-f-e} 
The observation that previously studied algorithms perform poorly even in terms of the weaker relaxation of stability implies that we need to look for new algorithms. So, now, a natural question that arises here is: \textit{Is there any hope at all for algorithms that are fair, PO, and stable?} Note that without the requirement of PO the answer to this question is a ``Yes''---at least when the fairness notion that is being considered is EF1, 
since it is easy to observe that the well-known \textit{draft} algorithm (where agents take turns picking their favourite good among the remaining goods) that is EF1 \cite[Sec.\ 3]{car16} is stable.  However, if we require PO, then we show that the answer to the question above is a ``No,'' in that there are no stable algorithms that always return an EF1 and even approximately-PO allocation. 

\begin{theorem} 
	Let $\mathcal{M}$ be an algorithm that is stable and always returns an EF1 allocation. Then, for any $\beta$ in $[1, \frac{T-3}{2}]$, $\mathcal{M}$ cannot be $\beta$-PO.
\end{theorem}

\begin{proof}
	Consider the instance in \Cref{tab3a} which represents agents' true valuation functions. For simplicity we assume that $T\mod 9 \equiv 0$. Next, since the agents are symmetric and it is easy to verify that in every EF1 allocation each of the agents have to get exactly two goods, let us assume w.l.o.g.\ that agent A receives $g_1$. Given this, now consider the instance in \Cref{tab3b} where agent B makes a mistake. Let us denote agent B's true utility function as $v$, misreport as $v' \in \frac{4T}{3}$-$N(v)$, and $(S_1, S_2)$ and $(S_1', S_2')$ as outcomes for the instances in Tables~\ref{tab3a}, \ref{tab3b}, respectively. From our discussion above, we know that $v(S_2) = \frac{4T}{9}$, since agent B gets two goods from the set $\{g_2, g_3, g_4\}$. 
	
		\begin{table}[tb]    
		\begin{subtable}{.5\linewidth}
			\centering       
			\begin{tabular}{l l|l}
				& A & B \\
				\cline{2-3}
				& & \\[-2ex]
				$g_1$ &  $\frac{T}{3}$ & $\frac{T}{3}$ \\[0.9ex]	  
				$g_2$ &  $\frac{2T}{9}$ & $\frac{2T}{9}$ \\[0.9ex]	
				$g_3$ &  $\frac{2T}{9}$ & $\frac{2T}{9}$ \\[0.9ex]	
				$g_4$ &  $\frac{2T}{9}$ & $\frac{2T}{9}$
			\end{tabular}
			\caption{Original instance} \label{tab3a}
		\end{subtable}%
		\begin{subtable}{.5\linewidth}
			\centering
			\begin{tabular}{l l|l}
				& A & B \\
				\cline{2-3}
				& & \\[-2ex]
				$g_1$ &  $\frac{T}{3}$ & $T-3$ \\[0.9ex]	  
				$g_2$ &  $\frac{2T}{9}$ & $1$ \\[0.9ex]
				$g_3$ &  $\frac{2T}{9}$ & $1$ \\[0.9ex]
				$g_4$ &  $\frac{2T}{9}$ & $1$
			\end{tabular}
			\caption{Instance where agent B makes a mistake} \label{tab3b}
		\end{subtable} 
	\end{table}
	
	Since $\mathcal{M}$ is stable, we know that $v(S_2) = v(S_2')$. Now, it is easy to 
	see that this is only possible if the set $S_2'$ has exactly two goods from $\{g_2, g_3, g_4\}$. So, let us assume w.l.o.g.\ that $S'_2 = \{g_2, g_3\}$. If 
	this is the case, then note that the allocation $(\{g_2, g_3, g_4\}, \{g_1\})$ Pareto dominates the allocation $(S_1', S_2')$ by a factor of $\frac{v'(g_1)}{v'(S_2')} = 
	\frac{T-3}{2}$, which in turn proves our theorem.
\end{proof}


Given this result and our observation in \Cref{sec:existalgo-stable} that previously studied algorithms perform poorly even in terms of weak-approximate-stability, it is clear that the best one can hope for is to design new algorithms that are fair, efficient, and approximately-stable. In \Cref{sec:rl} we show that this is possible. However, before we do that, in the next section we first present a necessary and sufficient condition for PMMS allocations when there are $n\geq 2$ agents.

\subsection{A necessary and sufficient condition for existence of PMMS allocations} \label{sec:neccsuff}
The general characterization result for PMMS allocations presented here will be useful in the next section to design an approximately-stable algorithm that produces a PMMS and PO allocation for the case of two agents. Additionally, we also believe that the result might potentially be of independent interest. 

For an agent $\ell\in[n]$, a set $Q \subseteq \mathcal{G}$, and a set $S \subseteq Q$, the result uses a notion of rank of $S$, denoted by $r_\ell^Q(S)$, and defined as the number of subsets of $Q$ that have value at most $v_\ell(S)$. More formally, 
\begin{equation} \label{def:rank}
	r_\ell^Q(S) = \bigl| \left\{P \mid P \subseteq Q, v_\ell(P) \leq v_\ell(S)\right\}\bigr|.
\end{equation}

The notion of rank has been previously considered in the fair division literature by \citet{rama13}. In particular,  according to our notation, they talk about $r_\ell^{\mathcal{G}}(S)$ in the context of fair division among two agents who have a strict preference orders over the subsets of $\mathcal{G}$, and one of their results is a weaker (since they assume a strict order over subsets of $\mathcal{G}$ which is not assumed here) version of the $n=2$ case of the theorem below \cite[Thm.\ 1(2)]{rama13}. 

\begin{theorem} \label{thm:pmmsiff}
	Given an instance with $m$ indivisible goods and $n \geq 2$ agents with additive valuation functions, an allocation $A = (A_1, \ldots, A_n)$ is a pairwise maximin share allocation if and only if $\forall i \in [n], \forall j \in [n]$, and $K_{ij} = A_i \cup A_j$,
	\begin{equation*}
		\min\left\{r_i^{K_{ij}}(A_i), \: r_j^{K_{ij}}(A_j)\right\} \geq 2^{\abs{K_{ij}}-1}. 
	\end{equation*}
\end{theorem}

Before we present a formal argument to prove this theorem, we present a brief overview. Overall, the proof uses some observations about the ranking function. In particular, for a set $Q$ and $S \subseteq Q$, one key observation is that the rank of $S$ is high enough (more precisely, greater than $2^{|Q|- 1}$) if and only if the value of this set is at least half that of $Q$. Once we have this, then the proof essentially follows by combining it with a few other observations about the ranking function and some simple counting arguments. 

\if\RankClaimsinApp0
More formally, we first state the following claims about the ranking function. The proofs of these directly follow from the way the ranking function is defined.
\else
More formally, we start by making the following claims about the ranking function. The proofs of the first two claims appear in \Cref{app:sec:neccsuff:proofs}.
\fi


\begin{restatable}{claim}{clmR} \label{clm:r}
	Let $Q \subseteq \mathcal{G}$ and $i \in [n]$ be some agent. Then, 
	\begin{enumerate}[label=\roman*),ref=\roman*]
		\item for $A \subseteq Q, B \subseteq Q$, $v_i(A) < v_i(B) \Leftrightarrow r_i^Q(A) < r_i^Q(B)$ \label{clm:r1}
		\item $r_i^Q(A) < r_i^Q(B) \Leftrightarrow r_i^{\mathcal{G}}(A) < r_i^{\mathcal{G}}(B)$.  \label{clm:r2}
	\end{enumerate} 
\end{restatable}

\if\RankClaimsinApp0
\begin{proof} To prove the first part, consider the sets $H_A = \{P \mid P \subseteq Q, v_i(P) \leq v_i(A) \}$ and $H_B = \{P \mid P \subseteq Q, v_i(P) \leq v_i(B) \}$. First, observe that $v_i(A) < v_i(B)$ if and only if $|H_A| < |H_B|$. Also, from the definition of the ranking function we know that $|H_A| < |H_B|$ if and only if $r_i^Q(A) < r_i^Q(B)$. Combining these we have our claim. 
	
	To prove the second part, observe that from the first part we know that $r_i^Q(A) < r_i^Q(B)$ if and only if $v_i(A) < v_i(B)$. Next, again using the first part with $Q = \mathcal{G}$, we have that $v_i(A) < v_i(B)$ if and only if $r_i^{\mathcal{G}}(A) < r_i^{\mathcal{G}}(B)$. Combining these we have our claim.
\end{proof}
\fi


\begin{restatable}{claim}{clmLevels} \label{clm:levels}
	Let $Q \subseteq \mathcal{G}$, $i \in [n]$ be some agent, $\ell \in \mathbb{Z}_{\geq 0}$, and $T_{\ell} = \{P \mid P \subseteq Q, r^Q_i(P) \leq \ell\}$. Then,
	\begin{enumerate}[label=\roman*),ref=\roman*]
		\item for $S \subseteq Q$, if $r_i^Q(S) = \ell$, then $|T_{\ell}| = \ell$ \label{clm:levels1}
		\item $|T_{\ell}| \leq \ell$. \label{clm:levels2}
	\end{enumerate} 
\end{restatable}

\if\RankClaimsinApp0
\begin{proof}
	To prove the first part, consider the set $H = \{P \mid P \subseteq Q, v_i(P) \leq v_i(S) \}$. First, note that from \Cref{clm:r}(\ref{clm:r1}) we can see that $H = T_\ell$. Next, from the definition of $r_i^Q(S)$, we know that $|H| = \ell$. Also, every element in $H$ will have a rank at most $\ell$ (since for each such set $S'$, $v_i(S') \leq v_i(S)$) and every element outside of $H$ will have a rank larger than $\ell$ (since for each such set $S'$, $v_i(S') > v_i(S)$). Hence, i) follows. 
	
	To prove the second part, consider the largest $\ell' \leq \ell$ such that there exists some $S \subseteq Q$ with $r_i^Q(S) = \ell'$. Now, from i) we know that the number of 
	subsets of $Q$ with rank at most $\ell'$ is exactly $\ell'$ and hence from our choice of $\ell'$ the statement follows. 
\end{proof}
\fi


\begin{restatable}{claim}{clmRankVal} \label{clm:rankVal}
	Let $A = (A_1, \ldots, A_n)$ be an allocation, and $i, j \in [n]$ be some agents. If $K_{ij} = A_i \cup A_j$, then 
	\begin{equation*}
		r^{K_{ij}}_i(A_i) > 2^{|K_{ij}|-1} \Leftrightarrow v_i(A_i) \geq \frac{v_i(K_{ij})}{2}.
	\end{equation*}
\end{restatable}

\begin{proof} \proofcase{$(\Rightarrow)$} Let us assume for the sake of contradiction that $\ell = r^{K_{ij}}_i(A_i) > 2^{|K_{ij}|-1}$, but $v_i(A_i) < \frac{v_i(K_{ij})}{2}$. Since $\ell > 2^{|K_{ij}|-1}$, we know from \Cref{clm:levels}(\ref{clm:levels1}) that there exists $S$ and $S^c = K_{ij} \setminus S$, such that $r^{K_{ij}}_i(S) \leq r^{K_{ij}}_i(A_i)$ and $r^{K_{ij}}_i(S^c) \leq r^{K_{ij}}_i(A_i)$. This in turn implies that using \Cref{clm:r}(\ref{clm:r1}) we have that $v_i(S) \leq v_i(A_i) < \frac{v_i(K_{ij})}{2}$ and $v_i(S^c) \leq v_i(A_i) < \frac{v_i(K_{ij})}{2}$, which is impossible since the valuation functions are additive. 
	
	\proofcase{$(\Leftarrow)$} Let $A$ be an allocation such that $v_i(A_i) \geq \frac{v_i(K_{ij})}{2}$, but $r^{K_{ij}}_i(A_i) \leq 2^{|K_{ij}|-1}$. Now, consider the set $H_1 = 
	\{P \mid P \subseteq K_{ij}, r^{K_{ij}}_i(P) > r^{K_{ij}}_i(A_i) \}$. From \Cref{clm:levels}(\ref{clm:levels1}) and using the fact that $r^{K_{ij}}_i(A_i) \leq 
	2^{|K_{ij}|-1}$ 
	we know that $|H_1| \geq 2^{\abs{K_{ij}}-1}$. Also, note that every set in $H_1$ has, by \Cref{clm:r}(\ref{clm:r1}), value greater than $v_i(A_i)$, which in turn implies that $A_j \notin H_1$. 
	Next, consider $H_2 = \{P^c \mid P \in H_1, P^c = K_{ij}\setminus P\}$, where additivity implies that every set in $H_2$ has value less than $v_i(A_1)$. Note that $\abs{H_1} + 
	\abs{H_2} \geq 2^{\abs{K_{ij}}}$ and $A_i$ is neither in $H_1$  nor $H_2$, which is impossible.
\end{proof}

Equipped with the claims above, we are now ready to prove our theorem.

\begin{proof}[Proof of \Cref{thm:pmmsiff}] \proofcase{$(\Rightarrow)$} Let us assume for the sake of contradiction that $A$ is a PMMS allocation and that there exists $i, j$ such that $\min\{r_i^{K_{ij}}(A_i), r_j^{K_{ij}}(A_j)\} < 2^{\abs{K_{ij}}-1}$. W.l.o.g.,  let us assume that $r_i^{K_{ij}}(A_i) < 2^{\abs{K_{ij}}-1}$. Next, consider the 
	set $H = \{B \mid B \subseteq K_{ij}, r^{K_{ij}}_i(B) > r^{K_{ij}}_i(A_i) \}$.  Since $r_i^{K_{ij}} (A_i) < 2^{\abs{K_{ij}}-1}$, we know 
	from \Cref{clm:levels}(\ref{clm:levels2}) that $|H| > 2^{\abs{K_{ij}}-1}$ (since there are $2^{\abs{K_{ij}}}$ subsets of $K_{ij}$). This 
	implies that there is a set $S$ and its complement $S^c = K_{ij} \setminus S$ such that,
	$r^{K_{ij}}_i(A_i) < \min\{r^{K_{ij}}_i(S),r^{K_{ij}}_i(S^c)\}$, which in turn using \Cref{clm:r}(\ref{clm:r1}) implies that 
	$v_i(A_i) < \min\{v_i(S), v_i(S^c)\}$. However, note that this contradicts the fact that $i$ has an MMS partition w.r.t.\ $j$ in $A$.  

	\proofcase{$(\Leftarrow)$} Let us assume for the sake of contradiction that there exists an agent $i$ such that $i$ does not have an MMS partition w.r.t.\ $j$, but 
	$\min\{r_i^{K_{ij}}(A_i), \allowbreak r_j^{K_{ij}}(A_j)\} \allowbreak \geq 2^{\abs{K_{ij}}-1}$. This implies that $v_i(A_i) < \frac{v_i(K_{ij})}{2}$, which in turn using \Cref{clm:rankVal} and the fact that $\min\{r_i^{K_{ij}}(A_i), \allowbreak r_j^{K_{ij}}(A_j)\} \allowbreak \geq 2^{\abs{K_{ij}}-1}$, implies that 
	$r_i^{K_{ij}}(A_i) = 2^{\abs{K_{ij}}-1}$. Next, since $i$ does not perceive $(A_i, A_j)$ to be an MMS partition w.r.t.\ $j$, there must exist a partition $(A_i', A_j')$ such that $A_i' \cup A_j' = K_{ij}$ and 
	$\min\{v_i(A'_i), v_i(A'_j)\} \allowbreak > v_i(A_i)$. This implies that, using \Cref{clm:r}(\ref{clm:r1}), we have that $r_i^{K_{ij}}(A_i') > r_i^{K_{ij}}(A_i)$ and $r_i^{K_{ij}}(A_j') > r_i^{K_{ij}}(A_i)$, or in other words that a set (i.e., $A_i'$) and its complement (i.e., $A_j' = K_{ij} \setminus A_i'$) both have rank greater than $r_i^{K_{ij}} (A_i)$. Now, if this is case, then one can see that, since $r_i^{K_{ij}}(A_i) = 2^{\abs{K_{ij}}-1}$, this implies there exists a set $S\subseteq K_{ij}$ and its complement $S^c = K_{ij} \setminus S$ such that $r_i^{K_{ij}}(S) \leq r_i^{K_{ij}}(A_i)$ and $r_i^{K_{ij}}(S^c) \leq r_i^{K_{ij}}(A_i)$. However, this is impossible because we can now use \Cref{clm:rankVal} to see that both $v_i(S)$ and $v_i(S^c)$ have value less than $\frac{v_i(K_{ij})}{2}$, which in turn contradicts the fact that $v_i$ is an additive valuation function. 
\end{proof}

In the next section we use this result to show an approximately-stable algorithm that is PMMS and PO when there are two agents. 

\subsection{rank-leximin: An approximately-stable PMMS and PO algorithm for two agents} \label{sec:rl}
The idea of our algorithm is simple. Instead of the well-known Leximin algorithm where one aims to maximize the minimum utility any agents gets, then the second minimum utility, and so on, our approach, which we refer to as the rank-leximin algorithm (\Cref{algo:rankLex}), is to do a leximin-like allocation, but based on the ranks of the bundles that the agents receive. Here for an agent $i$ and a bundle $B \subseteq \mathcal{G}$, by rank we mean $r_i^{\mathcal{G}}(B)$, as defined in (\ref{def:rank}). That is, the rank-leximin algorithm maximizes the minimum rank of the bundle that any agent gets, then it maximizes the second minimum rank, then the third minimum rank, and so on. Note that this in turn induces a comparison operator $\prec$ between two partitions and this is formally specified as rank-leximinCMP in \Cref{algo:rankLex}. Although the original leximin solution also returns a PMMS and PO allocation for two agents, recall that we observed in \Cref{sec:existalgo-stable} that it does not provide any guarantee even in terms of weak-approximate-stability, even when $\alpha = 4$. Rank-leximin on the other hand is PMMS and PO for two agents and, as we will show, is also $(2+\frac{12\alpha}{T})$-approximately-stable for all $\alpha \in (0, \frac{T}{3}]$. Additionally, it also returns an allocation that is PMMS and PO for any number of agents as long as they report ordinally equivalent valuation functions (see \Cref{def:equiv}).

\textbf{Remark:} Given the characterization result for PMMS allocations, one can come up with several algorithms that satisfy PMMS and PO. However, we consider rank-leximin here since it is a natural counterpart to the well-known leximin algorithm \cite{plaut18}. Additionally, it turns out that contrary to our initial belief the idea of rank-leximin is not new. \citet{rama13} considered it in the context of algorithms for two-agent fair division that satisfy Pareto-optimality, \textit{anonymity}, \textit{the unanimity bound}, and \textit{preference monotonicity} (see \cite[Sec.\ 3]{rama13} for definitions), where the \textit{unanimity bound} is a notion that one can show is equivalent to the notion of Maximin share (MMS) that is used in the computational fair division literature. So, with caveat that \citet{rama13} assumes that the agents have a strict preference orders over the subsets of $\mathcal{G}$ (which is not assumed here), the result of \citet[Thm.\ 2]{rama13} already proves that rank-leximin is PMMS and PO for two agents (MMS is equivalent to PMMS in the case of two agents), which is the result we show in \Cref{thm:main}. However, we still include our proof because of the indifference issue mentioned above, and since it almost follows directly from \Cref{thm:pmmsiff}.

\begin{algorithm}[tb]
	{\small \centering
		\noindent\fbox{%
			\begin{varwidth}{\dimexpr\linewidth-4\fboxsep-4\fboxrule\relax}
				\begin{algorithmic}[1]
					\small 
					\Procedurename rank-leximinCMP$(P, T)$
					\Input two partitions $P, T \in \Pi_n(\mathcal{G})$ 
					\Output returns true if $P \prec T$, i.e., if $P$ is before $T$ in the rank-leximin sorted order
					
					\State $R_P \leftarrow$ agents sorted in non-decreasing order of the rank of their bundles in $P$, i.e., based on $r^{\mathcal{G}}_i (P_i)$, with ties broken in 
					some arbitrary but consistent way throughout
					\State $R_T \leftarrow$ similar ordering as in $R_P$ above, but based on $r^{\mathcal{G}}_i (T_i)$
					
					\For{each $\ell \in [n]$}
					\State $i \leftarrow R^\ell_P$ \Comment{{\footnotesize$\ell$-th agent in $R_P$}}
					\State $j \leftarrow R^\ell_T$ \Comment{{\footnotesize$\ell$-th agent in $R_T$}}
					
					\If{$r^{\mathcal{G}}_i (P_i) \neq r^{\mathcal{G}}_j (T_j)$}
					\State \textbf{return} $r^{\mathcal{G}}_i (P_i) < r^{\mathcal{G}}_j (T_j)$
					\EndIf
					\EndFor
					\State \textbf{return} false
					\Statex  
					\Main
					\Input for each agent $i \in [n]$, their valuation function $v_i: 2^{\mathcal{G}} \to \mathbb{R}_{\geq 0}$  
					\Output an allocation $A = (A_1, \ldots, A_n)$ that is PMMS and PO  
					
					
					\State $\mathcal{L} \leftarrow$ perform a rank-Lexmin sort on $\Pi_n(\mathcal{G})$ based on the rank-leximinCMP operator defined above
					
					\State \textbf{return} $A = (A_1, \ldots, A_n)$ that is the last element in $\mathcal{L}$
				\end{algorithmic}
		\end{varwidth}}
		\caption{rank-leximin algorithm}
		\label{algo:rankLex}
	}
\end{algorithm}

Below we first show that the rank-leximin algorithm always returns an allocation that is PMMS and PO for the case of two agents. The fact that it is PO can be seen by  using some of the properties that we proved about the ranking function in the previous section, while the other property follows from combining \Cref{thm:pmmsiff} along with a simple pigeonhole argument. Following this, we also show that rank-leximin always returns a PMMS and PO allocation when there are $n\geq 2$ agents with ordinally equivalent valuation functions. The proof of this, which appears in \Cref{app:sec:rl:proofs}, is slightly more involved and it proceeds by first showing how rank-leximin returns such an allocation when all the agents have identical valuation functions. Once we have this, then the theorem follows by repeated application of \Cref{obs2}.

\begin{theorem}\label{thm:main}
	Given an instance with $m$ indivisible goods and two agents with additive valuation functions, the rank-leximin algorithm (\Cref{algo:rankLex}) 
	returns an allocation that is PMMS and PO.
\end{theorem}

\begin{proof}
	Let $A= (A_1,\ldots, A_n)$ be the allocation that is returned by the rank-leximin algorithm (\Cref{algo:rankLex}). Below we will first show that $A$ is PO and subsequently argue why it is PMMS.
	
	Suppose $A$ is not Pareto optimal. Then there exists another allocation $A'$ such that for all $i \in [n]$, $v_i(A'_i) \geq v_i(A_i)$, and the inequality is strict for at least one of the agents, say $j$. This in turn implies that using \Cref{clm:r} we have that for all $i \in [n]$, $r_i^{\mathcal{G}}(A'_i) \geq r_i^{\mathcal{G}} (A_i)$ and $r_j^{\mathcal{G}}(A'_j) > r_j^{\mathcal{G}} 
	(A_j)$. However, this implies that from the procedure rank-leximinCMP in \Cref{algo:rankLex} we have that $A \prec A'$, and this in turn directly contradicts the fact that $A$ was the allocation that was returned.
	
	To show that $A$ is PMMS, consider agent $1$ and all the sets $S$ such that $r_1^{\mathcal{G}}(S) \geq 2^{m-1}$. From \Cref{clm:levels} we know that there 
	are at least $2^{m-1} + 1$ such sets. Therefore, if $S^c = \mathcal{G} \setminus S$, then there is at least one $S$ such that $r_2^{\mathcal{G}}(S^c) \geq 2^{m-1}$. This implies $(S, S^c)$ is an allocation such that $\min\{r_1^{\mathcal{G}}(S), r_2^{\mathcal{G}}(S^c)\} \geq 2^{m-1}$. Now since rank-leximin maximizes the minimum rank that any agent receives, we have that $\min\{r_1^{\mathcal{G}}(A_1), r_2^{\mathcal{G}}(A_2)\} \geq 2^{m-1}$, and so now we can  use \Cref{thm:pmmsiff} to see that $A$ is a PMMS allocation.   
\end{proof}

\begin{restatable}{theorem}{thmRlnAgents} \label{thm:rl-nagents}
	Given an instance with $m$ indivisible goods and $n\geq2$ agents with additive valuation functions, the rank-leximin algorithm (\Cref{algo:rankLex}) returns an allocation $A$ that is PMMS and PO if all the agents report ordinally equivalent valuation functions.	
\end{restatable}

Now that we know rank-leximin produces a PMMS and PO allocation, we will move on to see how approximately-stable it is in the next section. However, before that we make the following remark. 


\textbf{Remark:} The rank-leximin algorithm takes exponential time. Note that this is not surprising since finding PMMS allocations is NP-hard even for two identical agents---one can see this by a straightforward reduction from the well-known Partition problem. 

\subsubsection{rank-leximin is \texorpdfstring{$\left(2+ \mathcal{O}(\frac{\alpha}{T})\right)$}{}-approximately-stable for \texorpdfstring{$\alpha \leq \frac{T}{3}$}{}} \label{sec:rl-stability}

\if\rlStabilityClaimsinApp0
Before we show how approximately-stable rank-leximin is, we prove the following claims which will be useful in order to prove our result. 

\else
Before we show how approximately-stable rank-leximin is, we state the following claims which will be useful in order to prove our result. The proofs of these claims appear in \Cref{app:sec:rl-st:proofs}.
\fi

\begin{restatable}{claim}{clmDiffvv} \label{clm:diffvv}
	Let $v$ be an additive utility function and for some $\alpha > 0$, let $v' \in \alpha$-$N(v)$. Then, for any $S_1, S_2 \subseteq{G}$,
	\begin{enumerate}[label=\roman*),ref=\roman*]
		\item $\abs{v(S_1) - v'(S_1)} \leq \frac{\alpha}{2}$ \label{clm:diffvvi}
		\item if $v(S_1) > v(S_2)$ and $v'(S_1) \leq v'(S_2)$, then $v(S_1) - v(S_2) \leq \alpha$. \label{clm:diffvvii}
	\end{enumerate}
\end{restatable}

\if\rlStabilityClaimsinApp0
\begin{proof}
	To prove the first part, suppose $\abs{v(S_1) - v'(S_1)} > \frac{\alpha}{2}$. Let $S_1^c = \mathcal{G} \setminus S_1$. Since $v, v'$ are additive, we know that $v(S_1) + v(S^c_1) = T$ (and similarly for $v'$), and so this in turn implies that $\abs{v'(S^c_1) - v(S^c_1)} > \frac{\alpha}{2}$. So, using this, we have,
	\begin{align*}
		\sum_{g \in \mathcal{G}} \abs{v(g) - v'(g)} &= \sum_{g \in S_1} \abs{v(g) - v'(g)} + \sum_{g \in S_1^c} \abs{v'(g) - v(g)}\\
		&\geq |{\textstyle\sum_{g \in S_1} (v(g) - v'(g)) }| + |{\textstyle\sum_{g \in S_1^c} (v'(g) - v(g))}|\\
		& = \abs{v(S_1) - v'(S_1)} + \abs{v'(S_1^c) - v(S_1^c)}\\
		&> \alpha,
	\end{align*}
	which is a contradiction since $v' \in \alpha$-$N(v)$. 
	
	To prove the second part, observe that from the first part we have,
	\begin{align*}
		v(S_1) - v(S_2) &\leq v'(S_1) + \frac{\alpha}{2} - (v'(S_2) - \frac{\alpha}{2})\\
		&= v'(S_1) - v'(S_2) + \alpha\\
		&\leq \alpha,
	\end{align*}
	where the last inequality follows from the fact that $v'(S_1) \leq v'(S_2)$. 
\end{proof}
\fi

\begin{restatable}{claim}{clmValbound} \label{clm:valbound}
	Given an instance with $m$ indivisible goods and two agents with additive valuation functions, let $(A_1, A_2)$ be a PMMS allocation. If for an agent $i \in [2]$, $k=|A_{j}|\geq 2$, where $j \neq i$, then
	\begin{enumerate}[label=\roman*),ref=\roman*]
		\item $v_i(A_i) \geq v_i(M_i^{j})$, where $ M_i^{j}$ is the maximum valued good of agent $i$ in the bundle $A_{j}$ \label{clm:valboundi}
		\item $v_i(A_{j}) - v_i(A_i) \leq v_i(m_i^{j})$, where $m_i^{j}$ is the minimum valued good of agent $i$ in the bundle $A_{j}$ \label{clm:valboundii}
		\item $v_i(A_{i}) \geq \frac{k-1}{2k-1} T$. \label{clm:valboundiii}
	\end{enumerate}
\end{restatable}

\if\rlStabilityClaimsinApp0
\begin{proof}
	To prove the first part, let us assume that $v_i(A_i) < v_i(M_i^{j})$. Since $k \geq 2$, we have another good $g \in A_j$ such that $g \neq M_i^{j}$. So, now, consider the allocation $(A_i \cup \{g\}, \{M_i^j\})$. Note that using additivity we have that $\min\{v_i(A_i \cup \{g\}), v_i(M_i^j)\} > v_i(A_i)$, which in turn contradicts the fact that $(A_i, A_j)$ is a PMMS allocation.
	
	To prove the second part, let us assume that $v_i(A_{j}) - v_i(A_i) > v_i(m_i^{j})$. Next, consider the allocation $(A_i \cup \{m_i^{j}\}, A_j\setminus\{m_i^j\})$. Note that using additivity we have that $\min\{v_i(A_i \cup \{m_i^{j}\}), v_i(A_j\setminus\{m_i^j\})\} > v_i(A_i)$, which in turn contradicts the fact that $(A_i, A_j)$ is a PMMS allocation.
	
	To prove the third part, observe that from the second part we know that $v_i(A_{i}) \geq v_i(A_j) - v_i(m_i^{j})$. Also, if $k=|A_j|$, then we can use the fact that the valuation functions are additive to see that $v_i(m_i^{j}) \leq \frac{v_i(A_j)}{k}$. So, using these, we have,
	\begin{align*} 
		v_i(A_{i}) &\geq v_i(A_j) - v_i(m_i^{j}) \geq v_i(A_j) - \frac{v_i(A_j)}{k} = (T-v_i(A_i))\left(\frac{k-1}{k}\right),
	\end{align*}
	where the last inequality follows from the fact that $v_i(A_i) + v_i(A_j) = T$. 
	
	Finally, rearranging the term above we have our claim. 
\end{proof}
\fi 

\begin{restatable}{claim}{clmMinrank} \label{clm:minrank}
	Given an instance with $m$ indivisible goods and two agents with additive valuation functions, let $(A_1, A_2)$ be the allocation that is returned by the rank-leximin algorithm for this instance. If $k_1 = 
	r^{\mathcal{G}}_1(A_1)$ 
	and $k_2 =r^{\mathcal{G}}_2(A_2)$, then $\min\{k_1, k_2\} \geq 2^{m-1}$ and $\max\{k_1, k_2\} > 2^{m-1}$.
\end{restatable}

\if\rlStabilityClaimsinApp0
\begin{proof}
	Since rank-leximin produces a PMMS allocation (\Cref{thm:main}), we know from \Cref{thm:pmmsiff} that $\min\{k_1, k_2\} \geq 2^{m-1}$. Also, if $\max\{k_1, 
	k_2\} \leq 2^{m-1}$, then using \Cref{clm:rankVal} we have that $v_1(A_1) < \frac{T}{2}$ and $v_2(A_2) < \frac{T}{2}$. However, this in turn contradicts the fact that rank-leximin is Pareto-optimal (\Cref{thm:main}) since swapping the bundles improves the utilities of both the agents.
\end{proof}
\fi 

Equipped with the claims above, we can now prove the approximate-stability bound for rank-leximin. The proof here proceeds by first arguing about how rank-leximin is stable with respect to equivalent instances. This is followed by showing upper and lower bounds on how weakly-approximately-stable it is, which in turn involves looking at several different cases and using several properties (some of them proved above and some that we will introduce as we go along) about the allocation returned by the rank-leximin algorithm.   

\begin{theorem} \label{thm:rlstability}
	Given an instance with $m$ indivisible goods and two agents with additive valuation functions, the rank-leximin algorithm is $(2+\frac{12\alpha}{T})$-approximately-stable for all $\alpha \leq \frac{T}{3}$.
\end{theorem}

\begin{proof}
	Let us consider two agents with valuation functions $v_1, v_2$. Since the rank-leximin is symmetric, we can assume w.l.o.g.\ that agent 1 is the one making a mistake. Throughout, let us denote this misreport by $v_1'$. Also, let rank-leximin$(v_1, v_2) = (A_1, A_2)$, and rank-leximin$(v'_1, v_2) = (S_1, S_2)$. Next, let us introduce the following notation that we will use throughout. For some arbitrary valuations $y_1, y_2$, if an allocation $(C_1, C_2)$ precedes an allocation $(B_1, B_2)$ in the rank-leximin order with respect to $(y_1, y_2)$ (i.e., according to the rank-leximinCMP function in \Cref{algo:rankLex}), then we denote this by $(B_1, B_2) \succ_{y_1, y_2} (C_1, C_2)$. In the case they are equivalent according to the rank-leximin operator, then we denote it by $(B_1, B_2) \equiv_{y_1, y_2} (C_1, C_2)$. Additionally, we use $(B_1, B_2) \succeq_{y_1, y_2} (C_1, C_2)$ to denote that $(C_1, C_2)$ either precedes or is equivalent to $(B_1, B_2)$.
	
	Equipped with the notations above, we now prove our theorem. To do this, first recall from the definition of approximate-stability (see \Cref{def:approx-stable}) that  we first need to show that if $v_1' \in \text{equiv}(v_1)$, then $v_1(A_1) = v_1(S_1)$. To see why this is true, consider the rank-leximinCMP function in \Cref{algo:rankLex} and observe that since $v_1' \in \text{equiv}(v_1)$, $v_1, v_1'$ have the same bundle rankings, and hence for any two partitions $P, T \in \pi_n(\mathcal{G})$, $P \prec_{v_1, v_2} T$ (i.e., $P$ appears before $T$ in the rank-leximin order) if and only if $P \prec_{v'_1, v_2} T$. This in turn along with the fact that rank-leximin uses a deterministic tie-breaking rule implies that rank-leximin$(v_1, v_2) = (A_1, A_2) = \text{rank-leximin}(v'_1, v_2)$, thus showing that it is stable with respect to equivalent instances. 
	
	Having shown the above, let us move to the second part where we show how weakly-approxmiately-stable rank-leximin is. For the rest of this proof, let $v_1' \in \alpha$-$N(v_i)$, for some $\alpha \in (0, \frac{T}{3}]$. Also, for $i \in [2]$, let the ranking function associated with valuation functions $v_i$ and $v_1'$ be $r^{\mathcal{G}}_i(\cdot)$ and $r_1^{'\mathcal{G}}(\cdot)$, respectively. Throughout, in order to keep the notations simple, we use $r_1(\cdot)$, $r'_1(\cdot)$, and $r_2(\cdot)$ to refer to $r^{\mathcal{G}}_1(\cdot)$, $r_1^{'\mathcal{G}}(\cdot)$, and $r^{\mathcal{G}}_2(\cdot)$, respectively. 
	
	In order to prove the bound on weak-approximate-stability, we show two lemmas. The first one shows an upper bound of 2 for the ratio $\frac{v_1(S_1)}{v_1(A_1)}$ and the second one shows a lower bound of $\frac{1}{2+\frac{12\alpha}{T}}$. 
	
	\begin{restatable}{lemma}{rlStUB} \label{lemma:rl-st-ub}
		For $\alpha > 0$, if $v_1' \in \alpha$-$N(v_1)$, rank-leximin$(v_1, v_2) = (A_1, A_2)$, and rank-leximin$(v'_1, v_2) = (S_1, S_2)$, then 
		\begin{equation*}
			\frac{v_1(S_1)}{v_1(A_1)} \leq 2.
		\end{equation*}
	\end{restatable}
	
	\if\rlStabilityUBlemmainApp1 
	\begin{proof}[Proof (sketch)] We proceed by considering three cases. Note that since rank-leximin returns a PMMS allocation (\Cref{thm:main}) we know from \Cref{thm:pmmsiff} that these are the only cases. Also, below we directly consider the case when $v_1(S_1) > v_1(A_1)$, since otherwise the bound trivially holds.
		
		\proofcase{Case 1.\ $r_1(A_1) > 2^{m-1}$:} In this case the bound follows by using \Cref{clm:rankVal}.
		
		\proofcase{Case 2.\ $r_1(A_1) = 2^{m-1}$ and $|S_1| = 1$:} In this case we proceed by first observing that since $|S_1| = 1$ and $v_1(A_1) < v_1(S_1)$, we have that $S_1 \subseteq A_2$. This in turn implies 	$v_1(S_1) \leq v_1(A_2)$, and so we have that $\frac{v_1(S_1)}{v_1(A_1)} \leq \frac{v_1(A_2)}{v_1(A_1)}$.
		Now, in order to upper-bound this, we further consider two cases. The first is when $|A_2| \geq 2$, and here the required bound follows by using \Cref{clm:valbound} from which we know that $v_1(A_1) \geq \frac{k-1}{2k-1} T$. The second case is when $|A_2| = 1$, and here, since $S_1 \subseteq A_2$, we have that $S_1 = A_2$. This implies {rank-leximin}$(v_1, v_2) = (A_1, A_2)$ and {rank-leximin}$(v'_1, v_2) = (A_2, A_1)$. Given this, one can make a series of observations to conclude that  $(A_1, A_2) \equiv_{v_1, v_2} (A_2, A_1)$ and $(A_1, A_2) \allowbreak \equiv_{v'_1, v_2} (A_2, A_1)$. However, this is a contradiction because rank-leximin breaks ties deterministically. 
		
		\proofcase{Case 3.\ $r_1(A_1) = 2^{m-1}$ and $|S_1| \geq 2$:} Using the facts that $v_1(S_1) > v_1(A_1)$, $r_1(A_1) = 2^{m-1}$, and $\text{rank-leximin}(v_1, v_2) = (A_1, A_2)$, it is not hard to see that for all $g \in S_1$, $v_1(S_1 \setminus \{g\}) \leq v_1(A_1)$.  This observation in turn can be used to show that  $\frac{v_1(S_1)}{v_1(A_1)} \leq \frac{|S_1|}{|S_1| - 1} \leq 2$. 
		
		Finally, combining all the three cases above gives us our lemma. As mentioned previously, the complete proof appears in \Cref{app:sec:rl-st-ub}. 
	\end{proof}
	\fi

	\if\rlStabilityUBlemmainApp0
	\begin{proof}
	To prove this, let us consider the following three cases. Note that since rank-leximin returns a PMMS allocation (\Cref{thm:main}) we know from 
	\Cref{thm:pmmsiff} that these are the only three cases. Also, note that the bound is trivially true when $v_1(A_1) - v_1(S_1) \geq 0$. So below we directly consider the case when $v_1(S_1) > v_1(A_1)$.
	
	\proofcase{Case 1.\ $r_1(A_1) > 2^{m-1}$:} Since $r_1(A_1) > 2^{m-1}$, we know using \Cref{clm:rankVal} that $v_1(A_1) \geq \frac{T}{2}$. Therefore, using this and the fact that $ v(S_1) \leq T$, we have, 
	\begin{align} \label{eqn:case1}
		&\frac{v_1(S_1)}{v_1(A_1)} \leq 2. 
	\end{align}
	
	\proofcase{Case 2.\ $r_1(A_1) = 2^{m-1}$ and $|S_1| = 1$:} First, observe that since $|S_1| = 1$ and $v_1(A_1) < v_1(S_1)$, we have that $S_1 \subseteq A_2$. This in turn implies 
	$v_1(S_1) \leq v_1(A_2)$ and so using this we have
	\begin{align} \label{eqn:ub}
		\frac{v_1(S_1)}{v_1(A_1)} \leq \frac{v_1(A_2)}{v_1(A_1)}.
	\end{align} 
	Now, in order to upper-bound (\ref{eqn:ub}), consider the following two cases.
	
	\begin{enumerate}[label=\roman*),ref=\roman*]
		\item \textbf{$|A_2| \geq 2$:}  Let $k = |A_2|$. Since $k\geq 2$, we know from \Cref{clm:valbound} that $v_1(A_1) \geq \frac{k-1}{2k-1} T$. So, using 
		this, we have
		\begin{align}  \label{eqn:case2}
			\frac{v_1(A_2)}{v_1(A_1)}  = \frac{T - v_1(A_1)}{v_1(A_1)} \leq  \frac{k}{k-1} \leq 2. 
		\end{align}
		
		\item \textbf{$|A_2| = 1$:} Since $|A_2| = 1$ and from above we know that $S_1 \subseteq A_2$, we have that $S_1 = A_2$. This in turn implies that 
		{rank-leximin}$(v_1, v_2) = (A_1, A_2)$ and {rank-leximin}$(v'_1, v_2) = (A_2, A_1)$. Next, consider the following observations.
		\begin{enumerate}[label=\alph*), ref=\alph*]
			\item $r_1(A_1) = 2^{m-1}$ (recall that this is the case we are considering)
			\item $r_2(A_2) = 2^{m-1} + 1$ (since $|A_2| = 1$, we know that $r_2(A_2) \leq 2^{m-1} + 1$, because valuation functions are additive and so any singleton set can have rank at most $2^{m-1} + 1$. However, $r_2(A_2) \geq 2^{m-1} + 1$, using \Cref{clm:minrank}.)
			
			\item $r'_1(A_1) = 2^{m-1}$ (since $|A_1| = m-1$, we know that $r'_1(A_1) \geq 2^{m-1}$. However, $r'_1(A_1) < 2^{m-1} + 1$, for if not then $(A_1, A_2) \succ_{v'_1, v_2} (A_2, 
			A_1)$, thus contradicting the fact that {rank-leximin}$(v'_1, v_2) = (A_2, A_1)$)
			\item $r_1(A_2) = 2^{m-1} + 1$ (since $|A_2| = 1 \Rightarrow r_1(A_2) \leq 2^{m-1} + 1$, and  $r_1(A_2) > r_1(A_1)$ since we are considering the case when $v_1(S_1) = v_1(A_2) > v_1(A_1)$)
			\item $r_2(A_1) = 2^{m-1}$ (since $|A_1| = m-1 \Rightarrow r_2(A_1) \geq 2^{m-1}$, and $r_2(A_1) < 2^{m-1} + 1$, for if not then $(A_2, A_1) \succ_{v_1, v_2} (A_1, 
			A_2)$, thus contradicting the fact that {rank-leximin}$(v_1, v_2) = (A_1, A_2)$)
			\item $r'_1(A_2) = 2^{m-1} + 1$ (since $|A_2| = 1 \Rightarrow r'_1(A_2) \leq 2^{m-1} + 1$, and $r_2(A_1) = 2^{m-1} \Rightarrow r'_1(A_2) \geq 2^{m-1} + 1$ using \Cref{clm:minrank})
		\end{enumerate}
		
		Now, combining all the observations above, we can see that $(A_1, A_2) \equiv_{v_1, v_2} (A_2, A_1)$ and $(A_1, A_2) \allowbreak \equiv_{v'_1, v_2} (A_2, A_1)$. However, note that this in turn is a contradiction because the rank-leximin algorithm breaks ties deterministically and so it couldn't have been the case that {rank-leximin}$(v_1, v_2) \allowbreak = (A_1, 
		A_2)$ and {rank-leximin}$(v'_1, v_2) = (A_2, A_1)$. Hence, this case is impossible.	
	\end{enumerate}
	
	\proofcase{Case 3.\ $r_1(A_1) = 2^{m-1}$ and $|S_1| \geq 2$:} First, note that since $v_1(S_1) > v_1(A_1)$, $r_1(A_1) = 2^{m-1}$, and $\text{rank-leximin}(v_1, v_2) = (A_1, A_2)$, we have that, for all $g \in S_1$, $v_1(S_1 \setminus \{g\}) \leq v_1(A_1)$. Why? Suppose if not, then consider the allocation $(S_1\setminus\{g\}), S_2 \cup \{g\})$, and observe that $(S_1\setminus\{g\}), S_2 \cup \{g\}) \succ_{v_1, v_2} (A_1, A_2)$, which in turn contradicts the fact that $\text{rank-leximin}(v_1, v_2) = (A_1, A_2)$. So, now, by adding up all the inequalities with respect all $g \in S_1$, we have that 
	\begin{align*}  \label{eqn:case3}
		(|S_1|-1) \cdot v_1(S_1) \leq |S_1| \cdot v_1(A_1). 
	\end{align*}
	This implies that we can use the fact that $|S_1| \geq 2$ to see that 
	\begin{align}
		&\frac{v_1(S_1)}{v_1(A_1)} \leq \frac{|S_1|}{|S_1| - 1} \leq 2. 
	\end{align}
	Finally, combining (\ref{eqn:case1}), (\ref{eqn:case2}), and (\ref{eqn:case3}), we have our lemma. 
\end{proof}
	\fi
	
	Next, we show a lower bound for $\frac{v_1(S_1)}{v_1(A_1)}$.
	
	\begin{restatable}{lemma}{rlStLB} \label{lemma:rl-st-lb}
		For $\alpha \in (0, \frac{T}{3}]$, if $v_1' \in \alpha$-$N(v_1)$, rank-leximin$(v_1, v_2) = (A_1, A_2)$, and rank-leximin$(v'_1, \allowbreak v_2) = (S_1, S_2)$, then 
		\begin{equation*}
			\frac{v_1(S_1)}{v_1(A_1)} \geq \frac{1}{2+\frac{12\alpha}{T}}.
		\end{equation*}
	\end{restatable}
	
	\if\rlStabilityLBlemmainApp0
	\begin{proof}
	Since we are trying to show a lower bound for $\frac{v_1(S_1)}{v_1(A_1)}$, below we consider the case when $v_1(A_1) > v_1(S_1)$, since otherwise it is trivially lower-bounded by 1. Also, note we can directly consider the case when $v_1(S_2) > v_1(S_1)$, for otherwise $v_1(S_1) = T - v_1(S_2) \geq \frac{T}{2}$ and so $\frac{v_1(S_1)}{v_1(A_1)} \geq \frac{1}{2}$. Therefore, throughout this proof, we have $v_1(S_2) > \frac{T}{2}$.
	
	Next, let us consider the following cases. Note that, by using \Cref{clm:minrank} with respect to the allocation $(S_1, S_2)$, these are the only three cases.
	
	\proofcase{Case 1.\ $r'_1(S_1) > 2^{m-1}$:} Since $r'_1(S_1) > 2^{m-1}$, we know from \Cref{clm:rankVal} that $v'_1(S_1) \geq \frac{T}{2}$. Also, since $v'_1 \in \alpha$-$N(v_1)$, from \Cref{clm:diffvv}(\ref{clm:diffvvi}) we have that $v_1'(S_1) - v_1(S_1) \leq |v_1'(S_1) - v_1(S_1)| \leq \frac{\alpha}{2}$. So, using these, we have, 
	\begin{align} \label{eqn:lb-case1}
		\frac{v_1(S_1)}{v_1(A_1)} \geq \frac{\frac{T}{2} - \frac{\alpha}{2}}{v_1(A_1)} \geq \frac{\frac{T}{2} - \frac{\alpha}{2}}{T} \geq \frac{1}{2 + \frac{4\alpha}{T}},
	\end{align}
	where the last inequality follows from the fact that $\alpha \in (0, \frac{T}{3}]$.
	
	\proofcase{Case 2.\ $r'_1(S_1) = 2^{m-1}$ and $\min\{r_1(A_1), r_2(A_2)\} = r_1({A_1})$:}  
	Since $\min\{r_1(A_1), r_2(A_2)\} = r_1({A_1})$ and rank-leximin$(v'_1, v_2) = (S_1, S_2)$, we have that $v_1'(A_1) \leq v_1'(S_1)$, for if not, then it is easy to see 	that $(A_1, A_2) \succ_{(v_1', v_2)} (S_1, S_2)$, thus contradicting the fact that rank-leximin$(v'_1, v_2) = (S_1, S_2)$. Also, since $v_1(A_1) > v_1(S_1)$ and from above we have $v_1'(A_1) \leq v_1'(S_1)$, from  \Cref{clm:diffvv}(\ref{clm:diffvvii}) we have know that $v_1(A_1) - v_1(S_1) \leq \alpha$. So, now, let us consider the following two cases.
	
	\begin{enumerate}[label=\roman*),ref=\roman*]
		\item \textbf{$|S_2| =1$:} First, note that since $|S_2|=1$ and $v_1(A_1) > v_1(S_1)$, we have that $S_2 \subseteq A_1$. This 
		in turn implies that $v_1(A_1) \geq v_1(S_2) > \frac{T}{2}$. So, now, using these we have, 
		\begin{align} \label{eqn:lb-case2.1}
			\frac{v_1(S_1)}{v_1(A_1)} \geq \frac{v_1(A_1) - \alpha}{v_1(A_1)} > 1 - \frac{2\alpha}{T}  \geq \frac{1}{2+\frac{4\alpha}{T}},
		\end{align} 
		where the last inequality follows from the fact that $\alpha \in (0, \frac{T}{3}]$.
		
		\item \textbf{$|S_2| \geq 2$:} Let $k = |S_2|$. Now, since $v_1(A_1) - 	v_1(S_1) \leq \alpha$, we have,
		\begin{align} \label{eqn:lb-case2.1.2}
			\frac{v_1(S_1)}{v_1(A_1)} \geq \frac{v_1(S_1)}{v_1(S_1) + \alpha} \nonumber &= 1- \frac{\alpha}{v_1(S_1) + \alpha} \nonumber \\
			&\geq 1- \frac{\alpha}{v'_1(S_1) - \frac{\alpha}{2} + \alpha} &&\text{{\scriptsize (using \Cref{clm:diffvv}(\ref{clm:diffvvi}))}} 
			\nonumber\\
			&\geq 1- \frac{\alpha}{\frac{k-1}{2k-1}T + \frac{\alpha}{2}}  &&\text{{\scriptsize $\left(\text{using \Cref{clm:valbound}(\ref{clm:valboundiii})}\right)$}} 
			\nonumber \\
			&\geq \frac{1}{2+\frac{4\alpha}{T}},
		\end{align}
		where the last inequality follows from the fact that $\alpha \in (0, \frac{T}{3}]$.
	\end{enumerate}
	
	\proofcase{Case 3.\ $r'_1(S_1) = 2^{m-1}$ and $\min\{r_1(A_1), r_2(A_2)\} = r_2({A_2})$:} First, note that we can directly consider the case when
	$\min\{r_1(A_1),\allowbreak r_2(A_2)\} = r_2(A_2) = 2^{m-1}$, for if not then we either have $v_1'(A_1) \leq v'_1(S_1)$, which can 
	be handled as in Case 2, or we have that $v_1'(A_1) > v'_1(S_1)$, which is impossible since this would imply $\min\{r'_1(A_1), r_2(A_2)\} > 
	2^{m-1} = r'_1(S_1) = \min\{r'_1(S_1), r_2(S_2)\}$, thus contradicting the fact that $\text{rank-leximin}(v'_1, v_2) = 
	(S_1,S_2)$. Additionally, since we are considering the case when $r_2(A_2) = 2^{m-1}$, we have from \Cref{clm:minrank} that $r_1(A_1) \geq 2^{m-1} + 1$, 
	which in turn implies, using \Cref{clm:rankVal}, that $v_1(A_1) \geq \frac{T}{2}$. 
	
	Given the observations above, let us now consider the following cases.
	
	\begin{enumerate}[label=\roman*),ref=\roman*]
		\item \textbf{$|S_2| \geq 2$ and $|S_1 \cap A_1| \geq 1$:} Let $k = |S_2|$ and let $X \subset \mathcal{G}$ such that $r_1(X) = 2^{m-1}$. Since $\allowbreak \text{rank-leximin}(v_1, \allowbreak v_2) = (A_1, A_2)$ and $\forall g \in \mathcal{G}$, and $\forall i \in [2]$, $v_i(g) > 0$, we have that $r_1(A_1\setminus\{g\}) \leq 2^{m-1}$, for if not then $\min\{r_1(A_1\setminus \{g\}), r_2(A_2 \cup \{g\})\} > 2^{m-1} =  r_2(A_2) = \min\{r_1(A_1), r_2(A_2)\}$, thus contradicting the fact that  $\text{rank-leximin}(v_1, v_2) = (A_1, A_2)$. So, next, let us consider a good $g \in S_1 \cap A_1$. Using the observations above, we have,
		\begin{align}\label{eqn:inter} 
			v_1(A_1) &\leq v_1(X) + v_1(g) && \text{{\scriptsize $\left(\text{since } r_1(A_1\setminus\{g\}) \leq 2^{m-1}\right)$}} \nonumber \\
			&\leq v_1(X) + v_1(S_1) && \text{{\scriptsize $\left(\text{since } v_1(g) \leq v_1(S_1)\right)$}} \nonumber\\
			&\leq 2v_1(S_1) + \alpha. 
		\end{align}
		The last inequality is true because either $v_1(S_1) \geq v_1(X)$, or, if not, then it is easy to see that, since $r_1'(S_1) = 2^{m-1}$ and $r_1(S_1) < r_1(X_1) = 2^{m-1}$, there exists an $H \subset \mathcal{G}$ such that $v_1(H) \geq v_1(X) > v_1(S_1)$ and $v'_1(H) \leq v'_1(S_1)$. This in turn implies, we can use \Cref{clm:diffvv}(\ref{clm:diffvvii}), to see that $v_1(H) - v_1(S) \leq \alpha$. 
		
		So, now, using (\ref{eqn:inter}), we have,
		\begin{align} \label{eqn:case2.2}
			\frac{v_1(S_1)}{v_1(A_1)} \geq \frac{1}{2} - \frac{\alpha}{2v_1(A_1)}
			&\geq \frac{1}{2} - \frac{\alpha}{T} &\text{{\scriptsize $\left(\text{since }v_1(A_1) \geq \frac{T}{2}\right)$}} \nonumber\\
			&\geq \frac{1}{2+\frac{12\alpha}{T}},
		\end{align}
		where the last inequality follows from the fact that $\alpha \in (0, \frac{T}{3}]$.
		
		\item \textbf{$|S_2| \geq 2$ and $|S_1 \cap A_1| = 0$:} Let $k = |S_2|$. First, note that since $|S_1 \cap A_1| = 0$, we have that $A_1 \subseteq S_2$. Therefore, using this, we have, 
		\begin{align} \label{eqn:case2.3}
			\frac{v_1(S_1)}{v_1(A_1)} \geq \frac{v_1(S_1)}{v_1(S_2)} &\geq \frac{v'_1(S_1) - \frac{\alpha}{2}}{v'_1(S_2) + \frac{\alpha}{2}} \nonumber \\
			&\geq \frac{\frac{k-1}{2k-1}T - \frac{\alpha}{2} }{\frac{k}{2k-1}T + \frac{\alpha}{2}} & \text{{\scriptsize (using \Cref{clm:valbound}(\ref{clm:valboundiii}))}} \nonumber \\
			&\geq \frac{1}{2 + \frac{12\alpha}{T}},
		\end{align}
		where the last inequality follows from the fact that $k \geq 2$ and $\alpha \in (0, \frac{T}{3}]$.
		
		\item\textbf{$|S_2| = 1$: } First, note that since $|S_2| = 1$, $r_2(S_2) \leq 2^{m-1} + 1$, as additivity implies that any singleton set can have rank at most $2^{m-1}+ 1$. Also, recall that we are considering the case when $r_1'(S_1) = r_2(A_2) = 2^{m-1}$. Now, since rank-lexmin$(v_1',v_2) = (S_1, S_2)$, we need to have $v_1'(S_1) \geq v_1'(A_1)$, for if not, then the fact that $r_2(A_2) = 2^{m-1}$ and $v'_1(A_1) > v'_1(S_1)$ implies that $(A_1, A_2) \succeq_{v_1, v_2} (S_1, S_2)$ and $(A_1, A_2) \succeq_{v'_1, v_2} (S_1, S_2)$. However, this is impossible since the rank-leximin algorithm uses a deterministic tie-breaking rule and hence if this is case, then rank-leximin$(v_1, v_2)$ and rank-leximin$(v'_1, v_2)$ cannot be different. Therefore, we have that $v_1'(S_1) \geq v_1'(A_1)$, which in turn can again be handled as in Case 2.    
	\end{enumerate} 
	
	Finally, combining (\ref{eqn:lb-case1})--(\ref{eqn:case2.3}), we have that $\frac{v_1(S_1)}{v_1(A_1} \geq 	\frac{1}{2+\frac{12\alpha}{T}}$. 
\end{proof} 
	\fi
	
	\if\rlStabilityLBlemmainApp1
	\begin{proof}[Proof (sketch)]
		Since we are trying to show a lower bound for $\frac{v_1(S_1)}{v_1(A_1)}$, below we consider the case when $v_1(A_1) > v_1(S_1)$, since otherwise it is trivially lower-bounded by 1. Also, note we can directly consider the case when $v_1(S_2) > v_1(S_1)$, for otherwise $v_1(S_1) = T - v_1(S_2) \geq \frac{T}{2}$ and so $\frac{v_1(S_1)}{v_1(A_1)} \geq \frac{1}{2}$. Next, we proceed by considering the following cases. Note that these are the only three cases using \Cref{clm:minrank} with respect to the allocation $(S_1, S_2)$.
		
		\proofcase{Case 1.\ $r'_1(S_1) > 2^{m-1}$:} Since $r'_1(S_1) > 2^{m-1}$, we know from \Cref{clm:rankVal} that $v'_1(S_1) \geq \frac{T}{2}$. Also, since $\|v_1-v_1'\|_{1} \leq \alpha$ and $v_1(\mathcal{G}) = v'_1(\mathcal{G}) = T$,  $v_1'(S_1) - v_1(S_1) \leq \frac{\alpha}{2}$. Combining these gives us our bound. 
		
		\proofcase{Case 2.\ $r'_1(S_1) = 2^{m-1}$ and $\min\{r_1(A_1), r_2(A_2)\} = r_1({A_1})$:}  
		Since $\min\{r_1(A_1), r_2(A_2)\} = r_1({A_1})$ and rank-leximin$(v'_1, v_2) = (S_1, S_2)$, we have that $v_1'(A_1) \leq v_1'(S_1)$, for if not, then it is easy to see 	that $(A_1, A_2) \succ_{(v_1', v_2)} (S_1, S_2)$, which is a contradiction. Also, since $v_1(A_1) > v_1(S_1)$ and $v_1'(A_1) \leq v_1'(S_1)$, from  \Cref{clm:diffvv}(\ref{clm:diffvvii}) we have that $v_1(A_1) - v_1(S_1) \leq \alpha$. Given these, one can show our bound by considering two cases. The first is when $|S_2| =1$, and here we can obtain our bound by observing that $S_2 \subseteq A_1$ and by using the fact that $v_1(A_1) > v_1(S_2)$. The second case is when $|S_2| \geq 2$, and here the bound follows by combining the fact $v_1(A_1) - v_1(S_1) \leq \alpha$ along with Claims~\ref{clm:diffvv}(\ref{clm:diffvvi}) and~\ref{clm:valbound}(\ref{clm:valboundiii}).
		
		\proofcase{Case 3.\ $r'_1(S_1) = 2^{m-1}$ and $\min\{r_1(A_1), r_2(A_2)\} = r_2({A_2})$:} Here we proceed by considering three sub-cases. The first one is when $|S_2| \geq 2$ and $|S_1 \cap A_1| = 0$, and here one can obtain the bound by combining \Cref{clm:valbound}(\ref{clm:valboundiii}) along with the observation that since $|S_1 \cap A_1| = 0$, $A_1 \subseteq S_2$. The second case is when $|S_2| = 1$, and here we argue that this can be handled by proceeding as in Case 2 above. The last one is when $|S_2| \geq 2$ and $|S_1 \cap A_1| \geq 1$. In this case the proof proceeds by first making the observation that $\forall g \in \mathcal{G}$, $r_1(A_1\setminus\{g\}) \leq 2^{m-1}$. This in turn can be used to show that $v_1(A_1) 	\leq 2v_1(S_1) + \alpha$, which can then be combined with \Cref{clm:diffvv}(\ref{clm:diffvvii}) to show the bound. 
		
		Finally, combining all the cases above gives us our lemma.  As mentioned previously, the complete proof appears in \Cref{app:sec:rl-st-lb}.
	\end{proof}
	\fi
	Since rank-leximin is stable with respect to equivalent instances, and combining the lemmas above we have that it is $(2+\frac{12\alpha}{T})$-weakly-approximately-stable for all $\alpha \leq \frac{T}{3}$, our theorem follows. 
\end{proof}

Note that for $\alpha \in (0, \frac{T}{3}]$, our approximate-stability bound is independent of the choice of $T$, whereas for all the previously studied algorithms that we had considered in \Cref{sec:existalgo-stable} it was dependent on $T$, and in fact was equal to $(T-1)$ even when $\alpha = 4$. Additionally, none of those algorithms are stable on equivalent instances either (i.e., the bound of $(T-1)$ was with respect to weak-approximate-stability and not the stronger notion of approximate-stability). Finally, it is interesting to note that although our model assumes that the misreport $v'$ always maintains the ordinal information (over the singletons) of the true valuation $v$ (see the second condition in \Cref{def:alphaN}), the proof above does not use this fact. This implies that rank-leximin is stronger than what the theorem indicates since it is $(2+\frac{12\alpha}{T})$-stable for all $\alpha \leq \frac{T}{3}$ and for all possible misreports that are within an $L_1$ distance of $\alpha$ from the original valuation.\footnotemark

\footnotetext{Seeing this, one might wonder ``why have the condition of maintaining the ordinal information over the singletons (i.e., the second condition in \Cref{def:alphaN}) then?'' The reason is, recall that our original goal was look for stable algorithms, and so this condition was a natural one to impose since, as mentioned in the introduction, one cannot hope for stable algorithms without assuming anything about the kind of mistakes the agents might commit. However, as mentioned in \Cref{sec:st-f-e}, even with this assumption designing stable algorithms turned out to be impossible, and hence we relaxed the strong requirement of stability to approximate-stability. And under this relaxed definition, the condition turns out to be not as crucial.}

\section{Weak-approximate-stability in fair division of indivisible goods} \label{sec:weak-stability}

In this section we consider weak-approximate-stability, the weaker notion of approximate stability defined in \Cref{def:weak-approx} and we show how a very simple change to the existing algorithms can ensure more stability. In particular, we show that any algorithm that returns a PMMS and PO allocation for two agents can be made $(4 + \frac{6\alpha}{T})$-weakly-approximately-stable for $\alpha \leq \frac{T}{3}$. The main idea needed to achieve this is to modify the algorithm that produces a PMMS and PO allocation and convert  it into a two-phase procedure where the first phase only handles instances where agent 1 has a good of value greater than $\frac{T}{2}$. In the case of such instances, we just output the allocation where agent 1 gets the highly-valued good and all the other goods are given to  agent 2. Below, in Algorithm~\ref{algo:mod}, we formally describe the framework and subsequently show how this very minor change leads to better stability overall. \if\weakStabilityinApp2 Due to space constraints, the proof appears in Appendix~\ref{app:sec:weak-st} \fi

\begin{algorithm}[tb]
	{\small\centering
		\noindent\fbox{%
			\begin{varwidth}{\dimexpr\linewidth-4\fboxsep-4\fboxrule\relax}
				\begin{algorithmic}[1]
					\small 
					\Input for each agent $i \in [2]$, their valuation function $v_i$ 
					\Output an allocation $A = (A_1, A_2)$ that satisfies the same properties as $\mathcal{M}$
					
					\State $g_1^{\max} \leftarrow \argmax_{g \in \mathcal{G}} v_1(g)$ \Comment{{\footnotesize in case of tie, pick the one with the lowest index}}
					\If{$v_1(g_1^{\max}) > \frac{T}{2}$} \Comment{\footnotesize  \textbf{Phase1}}
					\State \textbf{return} $A = \left({\{g_1^{\max}\}, \mathcal{G}\setminus\{g_1^{\max}\}}\right)$
					\EndIf 
					\State \textbf{return} $\mathcal{M}(v_1, v_2)$. \Comment{\footnotesize \textbf{Phase 2}}
				\end{algorithmic}
		\end{varwidth}}
		\caption{Framework for modified-$\mathcal{M}$}
		\label{algo:mod}
	}
\end{algorithm}

\begin{restatable}{theorem}{weakStThm} \label{thm:weak-st}
	Given an instance with $m$ indivisible goods and two agents with additive valuation functions, let $\mathcal{M}$ be an algorithm that returns a PMMS and PO allocation. Then, for $\alpha \leq \frac{T}{3}$, modified-$\mathcal{M}$, which refers to the change as described in Algorithm~\ref{algo:mod}, is $(4 + \frac{6\alpha}{T})$-weakly-approximately-stable, and returns an allocation that is PMMS and PO.  
\end{restatable}

\if\weakStabilityinApp0
\begin{proof}	
	We begin by arguing that if $\mathcal{M}$ produces a PMMS and PO allocation then so does modified-$\mathcal{M}$. To see this, note that we just need to show that the output in Phase 1 is PMMS and PO, since otherwise modified-$\mathcal{M}$ returns the same outcome as $\mathcal{M}$. The fact that it is PMMS is in turn easy to see because agent 1 is envy-free (since they get a good of value greater than $\frac{T}{2}$) and the other agent gets $m-1$ goods. As for PO, note that agent 1 gets only one good with value greater than $\frac{T}{2}$, so they cannot be made better-off without making agent 2 worse-off.
	
	To show the bound on  weak-approximate-stability, let us consider two agents with utility functions $v_1, v_2$. Next, for some $\alpha \in (0, \frac{T}{3}]$, we denote the misreport by agent $i \in [2]$ as $v_i'$, and let $(-i)$ denote the agent $(i+1)\mod 2$. Also, let rank-leximin$(v_1, v_2) = (A_1, A_2)$, and, depending on which agent $i \in [2]$ is misreporting, let 
	rank-leximin$(v'_1, v_{2}) = (S_1, S_2)$ or rank-leximin$(v_1, v'_{2}) = (S_1, S_2)$. We refer to $(v_1, v_2)$ as the true report, and, depending on which agent $i \in [2]$ is misreporting, refer to $(v'_1, v_{2})$ or $(v_1, v'_{2})$ as the misreport. Given an instance, we say that the instance is in `Phase 1' if agent 1 has a good of value greater than $\frac{T}{2}$  (this corresponds to the `if' statement in Algorithm~\ref{algo:mod}). Otherwise, we say that the instance is in `Phase 2.' Also, throughout, for $i\in[2]$, we use $g_i^{\max}$ to denote the highest valued good in $\mathcal{G}$ according to agent $i$ (if there are multiple goods with the highest value, pick the one with the lowest index). Recall that, since $v_i' \in \alpha$-$N(v_i)$, $g_i^{\max}$ is the same in $v_i$ and $v_i'$.
	
	Equipped with the notations above, let us now consider the following cases. Note that since the allocation returned is identical whenever an instance is in Phase 1, we only need to analyze these cases. 
	\proofcase{Case 1.\ True report is in Phase 1, misreport is in Phase 2:} First, since there is a change in phase due to misreporting and the choice of phase only depends on agent 1's valuation function, note that agent $1$ is the one misreporting. Given this, below we derive an upper and lower bound for the ratio $\frac{v_1(S_1)}{v_1(A_1)}$.
	
	Note that since $(v_1, v_2)$ in Phase 1, we have that $A_1 = \{g_1^{\max}\}$ and $v_1(g_1^{\max}) > \frac{T}{2}$. Therefore,
	\begin{equation} \label{eq:weak-arrox1}
		\frac{v_1(S_1)}{v_1(A_1)} = \frac{v_1(S_1)}{v_1(g_1^{\max})} \leq \frac{T}{v_1(g_1^{\max})} < 2.
	\end{equation}
	
	Next, to lower bound the ratio, let us consider the following two cases.
	\begin{enumerate}[label=\roman*),ref=\roman*]
		\item \textbf{$|S_2| = 1$:} Since $|S_2| = 1$, we have that $v_1'(S_1) \geq v_1'(\mathcal{G} \setminus \{g_1^{\max}\})$. Also, since $(v_1', v_{2})$ in Phase 2, we know that $v_1'(g_1^{\max}) \leq \frac{T}{2}$. So, now, using these and Claim~\ref{clm:diffvv}(\ref{clm:diffvvi}), we have
		\begin{align} \label{eq:weak-arrox2}
			\frac{v_1(S_1)}{v_1(A_1)} \geq \frac{v'_1(S_1) - \frac{\alpha}{2}}{v_1(A_1)} \geq \frac{v_1'(\mathcal{G} \setminus \{g_1^{\max}\}) - \frac{\alpha}{2}}{v_1(A_1)} 
			&\geq \frac{\frac{T}{2} - \frac{\alpha}{2}}{T} 
			\geq \frac{1}{2+\frac{4\alpha}{T}},
		\end{align}
		where the last inequality follows since $\alpha \in (0, \frac{T}{3}]$.
		
		\item \textbf{$|S_2| \geq 2$:} Let $k = |S_2|$. Since $k \geq 2$, we can use Claim~\ref{clm:diffvv}(\ref{clm:diffvvi}) and Claim~\ref{clm:valbound}(\ref{clm:valboundiii}), to see that 
		\begin{align} \label{eq:weak-arrox3}
			\frac{v_1(S_1)}{v_1(A_1)} &\geq \frac{v'_1(S_1) - \frac{\alpha}{2}}{v_1(A_1)}
			\geq \frac{\frac{k-1}{2k-1}T-\frac{\alpha}{2}}{T} 
			\geq \frac{1}{4 + \frac{6\alpha}{T}},
		\end{align}
		where the last inequality follows from the facts that $k \geq 2$ and $\alpha \in (0, \frac{T}{3}]$.  
	\end{enumerate}
	
	\proofcase{Case 2.\ True report is in Phase 2, misreport is in Phase 1:} Similar to the previous case, here again we have that agent 1 is the one misreporting. Also, since $(v'_1, v_2)$ in Phase 1, we have that $S_1 = \{g_1^{\max}\}$ and $v'_1(g_1^{\max}) > \frac{T}{2}$. Just like in the previous case, we again upper and lower bound the ratio $\frac{v_1(S_1)}{v_1(A_1)}$.
	
	For the upper bound, first note that since $(v_1, v_2)$ is in Phase 2, we have $v_1(g_1^{\max}) \leq \frac{T}{2}$. Next, note that we only need to consider the case when $g_1^{\max} \in A_2$ ,since otherwise the ratio is upper-bounded by 1. Also, if $|A_2| = 1$, then note that $v_1(A_1) \geq T - v_1(g_1^{\max}) \geq \frac{T}{2}$, and so if this is the case then again the ratio is upper-bounded by 1. Finally, if $k = |A_2| \geq 2$, then we can use Claim~\ref{clm:valbound}(\ref{clm:valboundiii}) to see that,
	\begin{align} \label{eq:weak-arrox4}
		\frac{v_1(S_1)}{v_1(A_1)} = \frac{v_1(g_1^{\max})}{v_1(A_1)} \leq \frac{\frac{T}{2}}{\frac{k-1}{2k-1}T} \leq \frac{3}{2}.
	\end{align} 	
	For the lower bound, since $v'_1(g_1^{\max}) > \frac{T}{2}$ and $\alpha \in (0, \frac{T}{3}]$, we can use \Cref{clm:diffvv}(\ref{clm:diffvvi}) to see that,
	\begin{align} \label{eq:weak-arrox5}
		\frac{v_1(S_1)}{v_1(A_1)} = \frac{v_1(g_1^{\max})}{v_1(A_1)} \geq \frac{\frac{T}{2} - \frac{\alpha}{2}}{T} \geq \frac{1}{2+\frac{4\alpha}{T}}.
	\end{align}
	\proofcase{Case 3.\ True report is in Phase 2, misreport is in Phase 2:}  Let $i\in[2]$ be the agent who is misreporting. First, for the lower bound, we can proceed exactly as in Cases 1(i) and 1(ii) above. For the upper bound, let us consider the three cases.
	\begin{enumerate}[label=\roman*),ref=\roman*]
		\item \textbf{$i = 1$ and $|A_{2}| = 1$:} Since $|A_{2}| = 1$, we know that $v_1(A_1) \geq v_1(\mathcal{G} \setminus \{g_i^{\max}\})$. Also, since $(v_1, v_2)$ is in Phase 2, $v_1(g_1^{\max}) \leq \frac{T}{2}$. So, using this, we have,
		\begin{align} \label{eq:weak-arrox6}
			\frac{v_1(S_1)}{v_1(A_1)} \leq \frac{T}{T-v(g_1^{\max})} \leq 2.
		\end{align}
		
		\item \textbf{$i = 2$ and $|A_{1}| = 1$:} First, note that $v_2(A_2) \geq \frac{T}{2}$, for, if otherwise, when the agents report $(v_1, v_2)$, $(A_2, A_1)$ Pareto dominates $(A_1, A_2)$, which is a contradiction. This is because, since $|A_1| = 1$ and we are in Phase 2, $v_1(A_1) \leq \frac{T}{2}$, and so the allocations can be swapped to (weakly) improve the utilities of both the agents. Therefore, using this, we have, 
		\begin{align} \label{eq:weak-arrox7}
			\frac{v_2(S_2)}{v_2(A_2)} \leq \frac{v_2(S_2)}{\frac{T}{2}} \leq 2.
		\end{align}
		
		\item \textbf{$|A_{-i}| \geq 2$:} Let $k = |A_{-i}|$. Since $k \geq 2$, we can use Claim~\ref{clm:valbound}(\ref{clm:valboundiii}), to see that
		\begin{align} \label{eq:weak-arrox8}
			\frac{v_i(S_i)}{v_i(A_i)}  \leq \frac{T}{\frac{k-1}{2k-1}T} \leq 3.
		\end{align}
	\end{enumerate}	
	Finally, combining all the observations, (\ref{eq:weak-arrox1})--(\ref{eq:weak-arrox8}), from above we have our theorem. 
\end{proof} 
\fi

\if\weakStabilityinApp1
\begin{proof}[Proof (sketch)]
	We begin by arguing that if $\mathcal{M}$ produces a PMMS and PO allocation then so does modified-$\mathcal{M}$. To see this, note that we just need to show that the output in Phase 1 is PMMS and PO, since otherwise modified-$\mathcal{M}$ returns the same outcome as $\mathcal{M}$. The fact that it is PMMS is in turn easy to see because agent 1 is envy-free (since they get a good of value greater than $\frac{T}{2}$) and the other agent gets $m-1$ goods. As for PO, note that agent 1 gets only one good with value greater than $\frac{T}{2}$, so they cannot be made better-off without making agent 2 worse-off.
	
	Given that  modified-$\mathcal{M}$ is guaranteed to produce a PMMS allocation, it just remains to show the bound on weak-approximate-stability. Note that the algorithm here is not symmetric, and hence we need to consider what happens when each agent misreports. However, except for this, the proof of this theorem has a similar flavour to that of Theorem~\ref{thm:rlstability}, in that this too involves a lot of case-by-case analysis. Therefore, in the interest of readability, we defer it to Appendix~\ref{app:sec:weak-st}.
\end{proof}
\fi

Note that the simple framework presented in Algorithm~\ref{algo:mod} can be used with respect to any algorithm $\mathcal{M}$ that is fair and efficient, and so one can derive, perhaps better bounds, when considering modified versions of specific algorithms (likes the ones in Section~\ref{sec:existalgo-stable}) that provide a PMMS, EFX, or EF1 guarantee. Therefore, although the overall change that is required to any algorithm is minor, the observation is important if we care about having stabler algorithms, since without this recall that we observed in Section~\ref{sec:existalgo-stable} that all the four previously studied algorithms are $(T-1)$-weakly-approximately-stable even when $\alpha = 4$. 

\section{Discussion}
The problem considered here was born as a result of an example we came across on the fair division website, Spliddit \cite{gold15}, where it seemed like arguably innocuous mistakes by an agent had significant ramifications on their utility for the outcome. Thus, the question arose as to whether such consequences were avoidable, or more broadly if we could have stabler algorithms in the context of fair allocations. In this paper we focused on algorithmic stability in fair and efficient allocation of indivisible goods among two agents, and towards this end, we introduced a notion of stability and showed how it is impossible to achieve exact stability along with even a weak notion of fairness like EF1 and even approximate efficiency. This raised the question of how to relax the strong requirement of stability, and here we proposed two relaxations, namely, approximate-stability and weak-approximate-stability, and showed how the previously studied algorithms in fair division that are fair and efficient perform poorly with respect to these relaxations. This lead to looking for new algorithms, and here we proposed an approximately-stable algorithm (rank-leximin) that guarantees a pairwise maximin share and Pareto optimal allocation, and presented a simple general framework for any fair algorithm to achieve weak-approximate stability. Along the way, we also provided a characterization result for pairwise maximin share allocations and showed how (in addition to the two-agent case) such allocations always exist as long as the agents report ordinally equivalent valuation functions. Overall, while the results demonstrate how one can do better in the context of two-agent fair allocations, our main contribution is in introducing the notion of stability and its relaxations, and in explicitly advocating for it to be used in the design of algorithms or mechanisms that elicit cardinal preferences. 

Moving forward, we believe that there is a lot of scope for future work. Especially when given the observation that humans may often find it hard to attribute exact numerical values to their preferences, we believe that some notion of stability should be considered when designing algorithms or mechanisms in settings where cardinal values are elicited. This opens up the possibility of asking similar questions like the ones we have in other settings where such issues may perhaps be more crucial. Additionally, specific to the problem considered here, there are several unanswered questions. For instance, the most obvious one is about the case when there are more than two agents. A careful reader would have noticed that all the algorithms we talked about here heavily relied on the fact that there were only two agents, and therefore none of them work when there are more. As another example---this one in the context of two-agent fair allocations---one question that could be interesting is to see if there are polynomial-time algorithms that are approximately-stable and can guarantee EF1/EFX and PO. Given the fact that polynomial time EF1 and PO algorithms exist for two agents, we believe that it would interesting to see if we can also additionally guarantee approximate-stability, or prove otherwise---in which case it would demonstrate a clear separation between unstable and approximately-stable algorithms. The rank-leximin algorithm presented here is approximately-stable and provides the stated guarantees in terms of fairness and efficiency, but it is exponential, and we could not answer this question either in the positive or otherwise. Therefore, this, too, remains to be resolved by future work.

\textbf{Acknowledgments. } We are grateful to Hong Zhou for several useful discussions. We also thank Nisarg Shah for comments on an early draft of this work. 
	
\printbibliography

\appendix 

\section{Omitted Proofs}

\subsection{Omitted Proofs from \texorpdfstring{\Cref{sec:prelims}}{Section~\ref{sec:prelims}}} \label{app:sec:prelims:proofs} 
\obsStable*

\begin{proof}
	The forward part (i.e., the `if' part) is obvious since by definition an algorithm is stable if it is $\alpha$-stable for all $\alpha > 0$.  To see the `only if' part, let us assume that $\mathcal{M}$ is $\alpha$-stable for some $\alpha > 0$. Next, for an arbitrary $\alpha' \neq \alpha$, let us consider an arbitrary agent $i_1 \in [n]$, an arbitrary $v_{i_1}$, an arbitrary $v_{i_1}' \in \alpha'$-$N(v_{i_1})$, and arbitrary reports of the other agents denoted by $v_{-i_{1}}$.
	
	First, note that for any $\alpha' < \alpha$, we have that 
	\begin{equation*}
	v_{i_1}' \in \alpha'\text{-}N(v_{i_1}) \Rightarrow v_{i_1}' \in \alpha\text{-}N(v_{i_1}),
	\end{equation*}
	which in turn implies that $\mathcal{M}$ is $\alpha'$-stable since it is $\alpha$-stable.
	
	Next, let us consider an $\alpha' > \alpha$. Now, in this case observe that we can construct a (finite) series of valuations functions $v_{i_2}, \cdots, v_{i_n}$, where for all $j \in [n]$, $v_{i_{j+1}} \in \alpha$-$N(v_{i_{j}})$ and where $v_{i_{n+1}} = v_{i_1}'$. This along with the fact that $\mathcal{M}$ is $\alpha$-stable implies that, for all $j \in [n]$, 
	\begin{equation*}
	v_{i_1}(\mathcal{M}(v_{i_j}, v_{-i_{1}})) =  v_{i_1}(\mathcal{M}(v_{i_{j+1}}, v_{-i_{1}})).
	\end{equation*}
	So, using the equalities above, we have, 
	\begin{equation*}
	v_{i_1}(\mathcal{M}(v_{i_1}, v_{-i_{1}})) = v_{i_1}(\mathcal{M}(v_{i_{n+1}}, v_{-i_{1}})) = v_{i_1}(\mathcal{M}(v'_{i_{1}},v_{-i_{1}})),
	\end{equation*}
	thus proving that $\mathcal{M}$ is $\alpha'$-stable. 
\end{proof}

%

\if\RankClaimsinApp1
\subsection{Omitted Proofs from \texorpdfstring{\Cref{sec:neccsuff}}{Section~\ref{sec:neccsuff}}} \label{app:sec:neccsuff:proofs}

\clmR*

\begin{proof} To prove the first part, consider the sets $H_A = \{P \mid P \subseteq Q, v_i(P) \leq v_i(A) \}$ and $H_B = \{P \mid P \subseteq Q, v_i(P) \leq v_i(B) \}$. First, observe that $v_i(A) < v_i(B)$ if and only if $|H_A| < |H_B|$. Also, from the definition of the ranking function we know that $|H_A| < |H_B|$ if and only if $r_i^Q(A) < r_i^Q(B)$. Combining these we have our claim. 
	
	To prove the second part, observe that from the first part we know that $r_i^Q(A) < r_i^Q(B)$ if and only if $v_i(A) < v_i(B)$. Next, again using the first part with $Q = \mathcal{G}$, we have that $v_i(A) < v_i(B)$ if and only if $r_i^{\mathcal{G}}(A) < r_i^{\mathcal{G}}(B)$. Combining these we have our claim.
\end{proof}

\clmLevels*

\begin{proof}
	To prove the first part, consider the set $H = \{P \mid P \subseteq Q, v_i(P) \leq v_i(S) \}$. First, note that from \Cref{clm:r}(\ref{clm:r1}) we can see that $H = T_\ell$. Next, from the definition of $r_i^Q(S)$, we know that $|H| = \ell$. Also, every element in $H$ will have a rank at most $\ell$ (since for each such set $S'$, $v_i(S') \leq v_i(S)$) and every element outside of $H$ will have a rank larger than $\ell$ (since for each such set $S'$, $v_i(S') > v_i(S)$). Hence, i) follows. 
	
	To prove the second part, consider the largest $\ell' \leq \ell$ such that there exists some $S \subseteq Q$ with $r_i^Q(S) = \ell'$. Now, from i) we know that the number of 
	subsets of $Q$ with rank at most $\ell'$ is exactly $\ell'$ and hence from our choice of $\ell'$ the statement follows. 
\end{proof}

%
\fi

\subsection{Omitted proofs from Section~\ref{sec:rl}} \label{app:sec:rl:proofs}
\thmRlnAgents*

\begin{proof}
	To prove this we first show how rank-Leximin always returns a PMMS and PO allocation $(A_1, \ldots, A_n)$ for agents with identical valuation functions. Once we have that, then the theorem follows by repeated application of \Cref{obs2} and by observing that the rank-leximin algorithm produces the same output for two instances that are equivalent. 
	
	First, note that we already know from the proof of \Cref{thm:main} that it is PO with respect to any instance with $n$ agents. So we just need to argue that it produces a PMMS allocation when the reports are identical. To see this, recall from \Cref{thm:pmmsiff} that we need to show that for all $i, j \in [n]$, $\min\{r_i^{Q}(A_i), r_j^{Q}(A_j)\} \geq 2^{k-1}$, where $Q = A_i \cup A_j$ and $k = |Q|$.
	
	Suppose this was not the case and there exists agents $i, j$ such that $r_i^{Q}(A_i) < 2^{k-1}$. Now, let us consider the set $H = \{S \mid S \subseteq Q \text{ and } r^Q_j(S) > r^Q_i(A_i) \}$.  Since $r_i^Q (A_i) < 2^{k-1}$, we know 
	from \Cref{clm:levels}(\ref{clm:levels2}) that $|H| > 2^{k-1}$ (since there are $2^k$ subsets of $Q$ in total). This in turn implies that, there is a set $S$ and its complement 	$S^c$, such that both $r^Q_j(S)$ and $r^Q_j(S^c)$  are greater than $r_i^Q (A_i)$. So, now, consider these sets $S$ and $S^c$, and ask 
	agent $i$ to pick the one she values the most. Let us assume without loss of generality that this is $S$. Note that since the valuations are additive and hence $v_i(S) \geq 	\frac{1}{2} v_i(Q) > v_i(A_i)$, we know from \Cref{clm:rankVal} that $r^Q_i(S) \geq 2^{k-1} > r^Q_i(A_i)$. This in turn implies that, 
	\begin{align*} \label{app:eqn}
		&\min \{r^Q_i(S), r^Q_j(S^c)\} > r_i^Q (A_i) \geq \min \{r_i^Q (A_i), r_j^Q 	(A_j) \} \nonumber \\
		\Rightarrow &\min \{r^Q_i(S), r^Q_i(S^c)\} > r_i^Q (A_i) \geq \min \{r_i^Q (A_i), r_i^Q	(A_j) \},
	\end{align*} 
	where the implication follows from the fact that $i$ and $j$ have identical valuation functions and so for any $S \subseteq \mathcal{G}$ and $A \subseteq S$, we have $r^S_i(A) = r^S_j(A)$.
	
	The observation above implies that we can use \Cref{clm:r}(\ref{clm:r2}) to see that,
	\begin{align*}
		&\min \{r^{\mathcal{G}}_i(S), r^{\mathcal{G}}_i(S^c)\} > \min \{r_i^{\mathcal{G}} (A_i), r_i^{\mathcal{G}} (A_j) \}\\
		\Rightarrow & \min \{r^{\mathcal{G}}_i(S), r^{\mathcal{G}}_j(S^c)\} > \min \{r_i^{\mathcal{G}} (A_i), r_j^{\mathcal{G}} (A_j) \},
	\end{align*}
	where the implication again follows from the fact that $i$ and $j$ have identical valuation functions. 
	
	Now, consider the allocation $A' = (A_1, \ldots, A_{i-1}, S, A_{i+1}, \ldots, A_{j-1}, S^c, A_{j+1}, \ldots, A_n)$. From the discussion above we know that $A \prec A'$ according to the rank-LeximinCMP operator in \Cref{algo:rankLex}. However, this contradicts the fact that $A$ was the allocation that was returned by the rank-Leximin algorithm.
\end{proof}

\if\rlStabilityClaimsinApp1
\subsection{Omitted proofs from Section~\ref{sec:rl-stability}} \label{app:sec:rl-st:proofs}

\clmDiffvv*

\begin{proof}
	To prove the first part, suppose $\abs{v(S_1) - v'(S_1)} > \frac{\alpha}{2}$. Let $S_1^c = \mathcal{G} \setminus S_1$. Since $v, v'$ are additive, we know that $v(S_1) + v(S^c_1) = T$ (and similarly for $v'$), and so this in turn implies that $\abs{v'(S^c_1) - v(S^c_1)} > \frac{\alpha}{2}$. So, using this, we have,
	\begin{align*}
		\sum_{g \in \mathcal{G}} \abs{v(g) - v'(g)} &= \sum_{g \in S_1} \abs{v(g) - v'(g)} + \sum_{g \in S_1^c} \abs{v'(g) - v(g)}\\
		&\geq |{\textstyle\sum_{g \in S_1} (v(g) - v'(g)) }| + |{\textstyle\sum_{g \in S_1^c} (v'(g) - v(g))}|\\
		& = \abs{v(S_1) - v'(S_1)} + \abs{v'(S_1^c) - v(S_1^c)}\\
		&> \alpha,
	\end{align*}
	which is a contradiction since $v' \in \alpha$-$N(v)$. 
	
	To prove the second part, observe that from the first part we have,
	\begin{align*}
		v(S_1) - v(S_2) &\leq v'(S_1) + \frac{\alpha}{2} - (v'(S_2) - \frac{\alpha}{2})\\
		&= v'(S_1) - v'(S_2) + \alpha\\
		&\leq \alpha,
	\end{align*}
	where the last inequality follows from the fact that $v'(S_1) \leq v'(S_2)$. 
\end{proof}

\clmValbound*

\begin{proof}
	To prove the first part, let us assume that $v_i(A_i) < v_i(M_i^{j})$. Since $k \geq 2$, we have another good $g \in A_j$ such that $g \neq M_i^{j}$. So, now, consider the allocation $(A_i \cup \{g\}, \{M_i^j\})$. Note that using additivity we have that $\min\{v_i(A_i \cup \{g\}), v_i(M_i^j)\} > v_i(A_i)$, which in turn contradicts the fact that $(A_i, A_j)$ is a PMMS allocation.
	
	To prove the second part, let us assume that $v_i(A_{j}) - v_i(A_i) > v_i(m_i^{j})$. Next, consider the allocation $(A_i \cup \{m_i^{j}\}, A_j\setminus\{m_i^j\})$. Note that using additivity we have that $\min\{v_i(A_i \cup \{m_i^{j}\}), v_i(A_j\setminus\{m_i^j\})\} > v_i(A_i)$, which in turn contradicts the fact that $(A_i, A_j)$ is a PMMS allocation.
	
	To prove the third part, observe that from the second part we know that $v_i(A_{i}) \geq v_i(A_j) - v_i(m_i^{j})$. Also, if $k=|A_j|$, then we can use the fact that the valuation functions are additive to see that $v_i(m_i^{j}) \leq \frac{v_i(A_j)}{k}$. So, using these, we have,
	\begin{align*} 
		v_i(A_{i}) &\geq v_i(A_j) - v_i(m_i^{j}) \geq v_i(A_j) - \frac{v_i(A_j)}{k} = (T-v_i(A_i))\left(\frac{k-1}{k}\right),
	\end{align*}
	where the last inequality follows from the fact that $v_i(A_i) + v_i(A_j) = T$. 
	
	Finally, rearranging the term above we have our claim. 
\end{proof}

\clmMinrank*

\begin{proof}
	Since rank-leximin produces a PMMS allocation (\Cref{thm:main}), we know from \Cref{thm:pmmsiff} that $\min\{k_1, k_2\} \geq 2^{m-1}$. Also, if $\max\{k_1, 
	k_2\} \leq 2^{m-1}$, then using \Cref{clm:rankVal} we have that $v_1(A_1) < \frac{T}{2}$ and $v_2(A_2) < \frac{T}{2}$. However, this in turn contradicts the fact that rank-leximin is Pareto-optimal (\Cref{thm:main}) since swapping the bundles improves the utilities of both the agents.
\end{proof}
\fi

\if\rlStabilityUBlemmainApp1
\subsubsection{Proof of Lemma~\ref{lemma:rl-st-ub}} \label{app:sec:rl-st-ub}
\rlStUB*

\fi

\if\rlStabilityLBlemmainApp1
\subsubsection{Proof of Lemma~\ref{lemma:rl-st-lb}} \label{app:sec:rl-st-lb}
\rlStLB*

\fi

\if\weakStabilityinApp1
\subsection{Omitted proofs from Section~\ref{sec:weak-stability}} \label{app:sec:weak-st}
\weakStThm*

\fi

\if\weakStabilityinApp2
\subsection{Omitted proofs from Section~\ref{sec:weak-stability}} \label{app:sec:weak-st}
\weakStThm*

\fi

\end{document}